\documentclass[acmsmall,screen,nonacm]{acmart}

\usepackage[british]{babel}
\usepackage{stmaryrd}
\usepackage{mathtools}
\usepackage{physics}
\usepackage{mathpartir}

\usepackage{quiver}
\usepackage{nicefrac}
\usepackage{xspace}
\usepackage{ebproof}
\usepackage{hyperref}
\usepackage[capitalize]{cleveref}
\usepackage{tikz}
\usepackage{kbordermatrix}
\usepackage{eqnarray}
\hypersetup{final}

\usepackage{tikzit}
\tikzstyle{gate}=[fill=white, draw=black, shape=rectangle, minimum height=0.43cm, minimum width=0.43cm, inner sep=0.1em]
\tikzstyle{control}=[fill=black, draw=black, shape=circle, scale=0.38]
\tikzstyle{not}=[shape=circle, path picture={ 
\draw[black](path picture bounding box.north) -- (path picture bounding box.south) (path picture bounding box.west) -- (path picture bounding box.east);
}, draw=black, scale=.8]
\tikzstyle{wcontrol}=[fill=white, draw=black, shape=circle, scale=0.35]
\tikzstyle{empty}=[fill=white, draw=black, shape=rectangle, inner sep=0.4em, emptyborder]
\tikzstyle{globalphase}=[fill=white, draw=black, inner sep=0.15em, shape=rounded rectangle, minimum height=0.4cm]
\tikzstyle{ancilla}=[fill=black, draw=black, shape=rectangle, minimum width=0.01cm, minimum height=0.25cm, inner sep=0.01em]
\tikzstyle{ground}=[fill=white, path picture={\draw[black](-1.5mm,0)--(-0.6mm,0);\draw[black,thick](-0.6mm,-1.75mm)--(-0.6mm,1.75mm) (0mm,-0.9mm)--(0mm,0.9mm) (0.6mm,-0.5mm)--(0.6mm,0.5mm);}, minimum width=0.1mm, draw=none, outer sep=0pt]
\tikzstyle{gate22}=[fill=white, draw=black, shape=rectangle, minimum height=.68cm, minimum width=0.6cm]
\tikzstyle{void}=[shape=rectangle, minimum height=0.5cm]
\tikzstyle{gate44}=[fill=white, draw=black, shape=rectangle, minimum height=1.43cm, minimum width=0.5cm]
\tikzstyle{divider}=[fill={rgb,255: red,220; green,220; blue,220}, draw=black, shape border rotate=90, regular polygon, regular polygon sides=3, inner sep=1.5pt, rounded corners=0.5mm]
\tikzstyle{gatherer}=[fill={rgb,255: red,220; green,220; blue,220}, draw=black, shape border rotate=-90, regular polygon, regular polygon sides=3, inner sep=1.5pt, rounded corners=0.5mm]
\tikzstyle{hyperedge}=[fill=white, draw=black, shape=rectangle, rounded corners=0.1cm, minimum height=.6cm, minimum width=.6cm]
\tikzstyle{square}=[fill=white, draw=black, shape=rectangle, minimum height=0.20cm, minimum width=0.20cm, inner sep=0.1em, thick]
\tikzstyle{gphase}=[rounded rectangle, rounded rectangle arc length=120, fill={zx_grey}, inner sep=2pt, font={\tiny\boldmath}, label distance=1mm, fill opacity=.8, text opacity=1, tikzit category=ZX]
\tikzstyle{customcontrol}=[fill=white, draw=black, inner sep=0.1em, shape=rounded rectangle, minimum height=0.2cm]

\tikzstyle{emptyborder}=[-, dash pattern=on 0.16em off 0.16em on 0.16em off 0.16em on 0.16em off 0em]
\tikzstyle{etc}=[-, draw=black, densely dashed, thick]
\tikzstyle{greyetc}=[-, draw={rgb,255: red,161; green,161; blue,161}, densely dashed, thick]
\tikzstyle{dots}=[-, dotted, draw=black, thick]
\tikzstyle{big}=[-, thick]
\tikzstyle{register}=[-, double]
\tikzstyle{grey}=[-, draw={rgb,255: red,161; green,161; blue,161}]
\tikzstyle{border}=[-, fill=white]

\input{figures/quantum-circuits.tikzdefs}
\newcommand{\tf}[1]{\scalebox{0.90}{\input{#1.tikz}}}

\theoremstyle{plain}
\newtheorem{theorem}{Theorem}

\newtheorem{lemma}[theorem]{Lemma}

\newtheorem{proposition}[theorem]{Proposition}

\theoremstyle{definition}
\newtheorem{definition}[theorem]{Definition}
\newtheorem{example}[theorem]{Example}
\newtheorem{remark}[theorem]{Remark}

\theoremstyle{definition}

\newcommand{\C}{\mathbb{C}}
\newcommand{\qreg}[1]{\mathsf{q}#1}
\newcommand{\qsem}[1]{\sem{\mathsf{q}#1}}
\newcommand{\unitary}[1]{\qreg{#1} \leftrightarrow \qreg{#1}}
\newcommand{\id}[1]{\mathsf{id}_{#1}}
\newcommand{\cat}[1]{\ensuremath{\mathbf{#1}}}
\newcommand{\sem}[1]{\ensuremath{\llbracket #1 \rrbracket}}
\newcommand{\dg}{^\dagger}
\newcommand{\spp}{\sem p^\perp}
\newcommand{\ifthen}[2]{\mathop{\mathsf{if\ let}}#1\mathbin{\mathsf{then}}#2}
\newcommand{\pattern}[2]{\mathsf{q}#1 < \mathsf{q}#2}
\newcommand{\Phase}[1]{\mathsf{Ph}(#1)}
\newcommand{\Gate}[1]{\mathsf{#1}\xspace}
\newcommand{\letin}[1]{\mathop{\mathsf{let}}#1\mathrel{\mathsf{in}}}
\DeclareMathOperator{\doubleplus}{+\kern-1ex+}
\newcommand{\iid}{\ensuremath{\mathrm{id}}}
\newcommand{\evalu}[2]{\mathsf{eval}^{\mathsf{u}}_{#1}(#2)}
\newcommand{\evalp}[2]{\mathsf{eval}^{\mathsf{p}}_{#1}(#2)}
\newcommand{\QFT}[1]{\ensuremath{\mathsf{QFT}_{#1}}}

\RequirePackage{tikz-cd}
\RequirePackage{amssymb}
\usetikzlibrary{calc}
\usetikzlibrary{decorations.pathmorphing}
\tikzset{curve/.style={settings={#1},to path={(\tikztostart)
    .. controls ($(\tikztostart)!\pv{pos}!(\tikztotarget)!\pv{height}!270:(\tikztotarget)$)
    and ($(\tikztostart)!1-\pv{pos}!(\tikztotarget)!\pv{height}!270:(\tikztotarget)$)
    .. (\tikztotarget)\tikztonodes}},
    settings/.code={\tikzset{quiver/.cd,#1}
        \def\pv##1{\pgfkeysvalueof{/tikz/quiver/##1}}},
    quiver/.cd,pos/.initial=0.35,height/.initial=0}
\tikzset{tail reversed/.code={\pgfsetarrowsstart{tikzcd to}}}
\tikzset{2tail/.code={\pgfsetarrowsstart{Implies[reversed]}}}
\tikzset{2tail reversed/.code={\pgfsetarrowsstart{Implies}}}
\tikzset{no body/.style={/tikz/dash pattern=on 0 off 1mm}}

\usepackage[obeyDraft]{todonotes}
\todostyle{alex}{color=cyan}
\todostyle{louis}{color=magenta}
\todostyle{kim}{color=green}
\todostyle{malin}{color=yellow}

\usepackage{ifdraft}
\ifdraft{
\usepackage{showlabels}
\paperwidth=\dimexpr \paperwidth + 6cm\relax
\oddsidemargin=\dimexpr\oddsidemargin + 3cm\relax
\evensidemargin=\dimexpr\evensidemargin + 3cm\relax
\marginparwidth=\dimexpr \marginparwidth + 3cm\relax
}{}

\title{Quantum circuits are just a phase}

\author{Chris Heunen}
\orcid{0000-0001-7393-2640}
\affiliation{
  \department{School of Informatics}
  \institution{University of Edinburgh}
  \city{Edinburgh}
  \postcode{EH8 9AB}
  \country{United Kingdom}
}
\email{chris.heunen@ed.ac.uk}

\author{Louis Lemonnier}
\orcid{0000-0003-1761-3244}
\affiliation{
  \department{School of Informatics}
  \institution{University of Edinburgh}
  \city{Edinburgh}
  \postcode{EH8 9AB}
  \country{United Kingdom}
}
\email{louis.lemonnier@ed.ac.uk}

\author{Christopher McNally}
\orcid{0000-0002-4927-0613}
\affiliation{
  \department{Center for Quantum Engineering}
  \institution{Massachusetts Institute of Technology}
  \city{Cambridge}
  \postcode{MA 02139-4307}
  \country{United States}
}
\email{mcnallyc@mit.edu}

\author{Alex Rice}
\orcid{0000-0002-2698-5122}
\affiliation{
  \department{School of Informatics}
  \institution{University of Edinburgh}
  \city{Edinburgh}
  \postcode{EH8 9AB}
  \country{United Kingdom}
}
\email{alex.rice@ed.ac.uk}

\begin{document}

\begin{abstract}
  Quantum programs today are written at a low level of abstraction---quantum circuits akin to assembly languages---and the unitary parts of even advanced quantum programming languages essentially function as circuit description languages.
  This state of affairs impedes scalability, clarity, and support for higher-level reasoning.
  More abstract and expressive quantum programming constructs are needed.

  To this end, we introduce a simple syntax for generating unitaries from ``just a phase''; we combine a (global) phase operation that captures phase shifts with a quantum analogue of the ``if let'' construct that captures subspace selection via pattern matching.
  This minimal language lifts the focus from gates to eigen\-decomposition, conjugation, and controlled unitaries; common building blocks in quantum algorithm design.

  We demonstrate several aspects of the expressive power of our language in several ways.
  Firstly, we establish that our representation is universal by deriving a universal quantum gate set.
  Secondly, we show that important quantum algorithms can be expressed naturally and concisely, including Grover's search algorithm, Hamiltonian simulation, Quantum Fourier Transform, Quantum Signal Processing, and the Quantum Eigenvalue Transformation.
  Furthermore, we give clean denotational semantics grounded in categorical quantum mechanics.
  Finally, we implement a prototype compiler that efficiently translates terms of our language to quantum circuits, and prove that it is sound with respect to these semantics.
  Collectively, these contributions show that this construct offers a principled and practical step toward more abstract and structured quantum programming.
\end{abstract}

\maketitle

\section{Introduction}

Quantum computers can accommodate algorithms that solve certain classes of problems exponentially faster than the best known classical algorithms~\cite{nielsenchuang}.
Spurred on by this promise, quantum hardware has developed to the point where it's now a commercial reality. The current state of the art is still modest---qubit counts in the hundreds, coherence times in the microseconds, and gate error rates in the hundredths of a percent---but capabilities keep advancing at pace~\cite{mckinsey}.

% The problem

As quantum computing hardware keeps developing, the bottleneck to useful application is increasingly shifting to quantum software development.
One reason for this lag is the low level of abstraction at which quantum computers are currently programmed.
Most quantum software today is written in terms of quantum circuits (see \cref{fig:toffoli,fig:grover} for an example)---or worse, using hardware-specific execution instructions.
These representations suffice for small-scale experimentation and applications, but face several challenges in the longer term\todo[malin]{Suggestion: for more complex systems in the longer term}:
\begin{itemize}
  \item \emph{Scalability.} Proven classical software engineering practice and principles show that developing and maintaining programs at larger scales needs functionality supporting modularity~\cite{sullivanetal:modularity}.
  A related challenge is that interfacing with existing classical (high-performance) infrastructure similarly requires more structured representations.

  \item \emph{Automatability.} Empirical evidence shows that the vast majority of quantum circuits implementing useful quantum algorithms is taken up by `bookkeeping'~\cite{valironetal:quipper}. The burden of having to write this boilerplate code can be shifted from the programmer to automated support systems when the representation has a high enough level of abstraction.\todo[kim]{How do we tackle this?}

  \item \emph{Understandability}. At the root of these challenges lies the problem that quantum circuits are too fine-grained for human programmers to understand quantum algorithms at a natural level. Intuition for quantum algorithms comes from quantum information theory, and ultimately linear algebra. Having to translate this in terms of specific quantum gate sets is merely an obscuring step.
  According to the empirically supported weak Sapir-Whorf hypothesis from cognitive linguistics, a language's structure influences a speaker's ability to perceive ideas, without strictly limiting or obstructing them~\cite{ahearn:sapirwhorf}. Having a more abstract principled representation of quantum programs can therefore aid in the discovery of new quantum algorithms.
\end{itemize}
Additionally---but we will not address this challenge explicitly here---to enable optimising compiler passes, it is helpful to start at a higher level of abstraction, so that as much as possible of the programmer's intent is retained~\cite{hoos:prematureoptimisation}. At higher levels of abstraction, different optimisations become apparent. For example, the Toffoli gate is easily seen to commute with \(\Gate{Z}\) measurements, but this property is obscured when it is expressed in terms of hardware-native gates as in \cref{fig:toffoli}.
To better support reasoning, optimisation, program synthesis and verification, principled quantum programming needs a representation with more expressive power\todo[malin]{Is expressive the right word} and structural clarity.
Working primarily with quantum circuits is \emph{just a phase} in the coming of age of quantum programming.

\begin{figure}
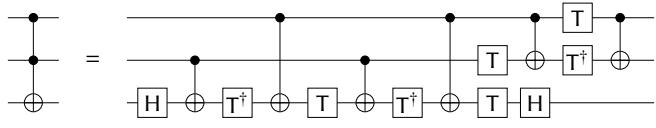

%   \[
%     \begin{quantikz}[row sep=2.25mm, column sep=3mm]
%       &\ctrl{1}&&\setwiretype{n}&&\setwiretype{q}&          &          &               & \ctrl{2} &          &          &               & \ctrl{2} &          & \ctrl{1} & \gate{\Gate{T}\phantom{^\dag}}      & \ctrl{1} & \\
%       &\ctrl{1}&&=
%       \setwiretype{n}&&\setwiretype{q}&          & \ctrl{1} &               &          &          & \ctrl{1} &               &          & \gate{\Gate{T}\vphantom{^\dag}} & \targ{}  & \gate{\Gate{T}^\dag} & \targ{}  & \\
%       &\targ{}&&\setwiretype{n}&&\setwiretype{q}& \gate{\Gate{H}} & \targ{}  & \gate{\Gate{T}^\dag} & \targ{}  & \gate{\Gate{T}\vphantom{^\dag}} & \targ{}  & \gate{\Gate{T}^\dag} & \targ{}  & \gate{\Gate{T}\vphantom{^\dag}} & \gate{\Gate{H}\vphantom{^\dag}} &               &          &
%     \end{quantikz}
%   \]
	\[
		\tf{toffoli}
	\]
  \caption{The Toffoli gate expanded into a quantum circuit of native gates. The expansion obscures that the operation commutes with \(\Gate{Z}\) measurement.}
  \label{fig:toffoli}
\end{figure}

% The state of the art

Unfortunately, quantum programming requires different abstractions to classical programming; existing classical constructs do not transfer cleanly. Conditional if-then-else constructions need care within quantum computing~\cite{badescupanangaden:qif,bisioetal:qif}. More generally, the quantum setting allows causal constructs fundamentally incompatible with classical control flow~\cite{chiribellaetal:switch,procopioetal:orders}. Control structures such as for and while loops are limited because they cannot inspect the quantum variable controlling the loop without altering its value~\cite{sabryvalironvizzotto:control,andresmartinezheunen:while}. The no-cloning theorem makes practical implementations of recursion schemes over quantum states very difficult~\cite{wootterszurek:nocloning,ying:foundations,zhangying:recursive}. Similarly, there are foundational challenges to higher-order structure~\cite{selinger:higherorder,paganietal:higherorder}.

As a result, the abstractions available in current quantum programming languages are limited and broadly fit into two categories:
\begin{itemize}
\item Many languages and libraries directly describe circuits~\cite{heimetal:quantumprogramming, valironetal:quipper}. These languages often have various quantum gates directly as primitives, often adding orthogonal features such as classical control~\cite{protoquipper,svoreetal:qsharp,crossetal:openqasm} or uncomputation~\cite{bichseletal:silq,hirataheunen:qurts}.
\item Other languages take a much larger departure from the circuit model. Some are based on reversible computing~\cite{caretteetal:sqrtpi,information-effects}, including languages which utilise symmetric pattern matching~\cite{sabryvalironvizzotto:control,phd-louis,daveetal:control} to define unitary operations. Alternative approaches~\cite{botoforslund:zeta} build on models of quantum computing such as the ZX-calculus~\cite{DBLP:conf/icalp/CoeckeD08}. Although these languages offer different abstractions over the circuit model, it is unclear how (or known to be hard~\cite{debeaudrap_et_al:LIPIcs.ICALP.2022.119}) to compile them down to circuits.
\end{itemize}

In this paper, we position ourselves between these two settings, offering an abstraction over quantum circuits while retaining a linear time compilation algorithm. We do this by introducing a new quantum programming construct, a quantum analogue of the ``if let'' statement, as used in Rust~\cite{klabniketal:rust}, which subsumes common operations such as conjugation and controlled blocks. When combined with an explicit treatment of global phase, often primitive quantum gates can be derived.
It leverages the fact that many quantum algorithms, and in fact many linear algebra techniques, have at their heart a decomposition into eigenspaces and a manipulation of eigenvectors; the expression
\[
  \ifthen{p}{e}
\]
represents the former, while the latter is captured by the global phase operator
\[
  \Phase{\theta}\text.
\]
The pattern $p$ specifies a case split by selecting a subspace of a variable's state space, with the ``if let'' expression applying its body to this subspace. Crucially, however, this subspace is not limited to align to classical values like $\ket{0}$ or $\ket{1}$, but also quantum values like $\ket{+}$ or higher-dimensional subspaces.
Our syntax consists of \emph{just a phase}, and is a simple but useful expression of the essence of eigendecompositions from linear algebra in quantum programming. We argue this in four ways.

First, this construct is expressive enough to serve as a foundational abstraction.
For example, conventional gate-level operations, that are usually taken as primitive, instead emerge as derived constructs. Here is a standard computationally universal gate set in the combinator language introduced in \cref{s:syntax}:
\begin{equationarray*}{rrlcrrl}
    \Gate{Z} &\coloneq& \ifthen{\ket{1}}{\Phase{\pi}} & \qquad & \Gate{X} &\coloneq& \ifthen{\ket{-}}{\Phase{\pi}}\\
    \Gate{T} &\coloneq& \ifthen{\ket{1}}{\Phase{\nicefrac{\pi}{4}}} & \qquad & \Gate{Y} &\coloneq& \ifthen{\Gate{S} \cdot \ket{-}}{\Phase{\pi}}\\
    \Gate{H} &\coloneq& \ifthen{\Gate{Y}^{\nicefrac{1}{4}} \cdot \ket{1}}{\Phase{\pi}}& \qquad & \Gate{CX} &\coloneq& \ifthen{\ket{1} \otimes \id{1}}{\Gate{X}}
\end{equationarray*}

\begin{figure}
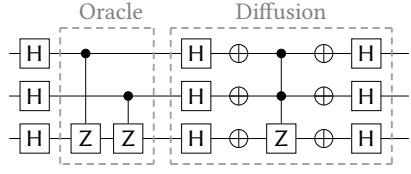

  % From https://github.com/Qiskit/textbook/blob/main/notebooks/ch-algorithms/grover.ipynb
%   \[
%     \begin{quantikz}[row sep=2mm, column sep=4mm]
%     \qw & \gate{\Gate{H}} &
%     \gategroup[3,steps=3,style={dashed,rounded corners, inner sep=0pt,color=gray,xshift=1mm},background,label style={label position=above,anchor=north,yshift=2.5mm, color=gray, scale=0.9}]{Oracle}
%       \qw & \ctrl{2} & \qw &
%     \gategroup[3,steps=6,style={dashed,rounded corners, inner sep=0pt,color=gray,xshift=1mm},background,label style={label position=above,anchor=north,yshift=2.5mm, color=gray, scale=0.9}]{Diffusion}
%     \qw & \gate{\Gate{H}} & \targ{} & \ctrl{1} & \targ{} & \gate{\Gate{H}} & \qw & \\
%       \qw & \gate{\Gate{H}} & \qw & \qw & \ctrl{1} & \qw & \gate{\Gate{H}} & \targ{} & \ctrl{1} & \targ{X} & \gate{\Gate{H}} & \qw & \\
%       \qw & \gate{\Gate{H}} & \qw & \gate{\Gate{Z}} & \gate{\Gate{Z}} & \qw & \gate{\Gate{H}} & \targ{} & \gate{\Gate{Z}} & \targ{} & \gate{\Gate{H}} & \qw &
%     \end{quantikz}
%   \]
	\[
		\tf{grover}
	\]
  \caption{A circuit for an instance of Grover's algorithm searching for two marked bit strings  011 and 101~\cite{figgattetal:grover}. The intent of the programmer (and the meaning of the program) is obfuscated by the circuit representation.}
  \label{fig:grover}
\end{figure}

As a second argument, we show that this language captures a broad class of quantum algorithms. For example, Grover's search algorithm (whose circuit representation is given in \cref{fig:grover}) iterates two main subroutines. The most important one, the diffusion operator
\begin{equation*}
  \Phase{\pi} \otimes \id n; \ifthen{\ket{+} \otimes \cdots \otimes \ket{+}}{\Phase{\pi}}
\end{equation*}
is a one-liner.
The programmer has to supply the oracle operator
\begin{equation*}
  \ifthen{\ket{\omega_1} \otimes \cdots \otimes \ket{\omega_n}} {\Phase{\pi}}\text,
\end{equation*}
where $\omega_j$ is the $j$\textsuperscript{th} bit in the binary expansion of the marked element,
which also simplifies.
In a similar way, we show that our representation can succinctly express important quantum algorithms including Quantum Fourier Transform, Hamiltonian simulation, Quantum Signal Processing, and Quantum Eigenvalue Transformation.

Third, we validate the practicality of our approach through a prototype compiler that translates our higher-level constructs into standard quantum circuits efficiently, which we present as an evaluation function converting any term to a canonical ``circuit-like'' form. This may also be regarded as an operational semantics for the language.

Fourth, we equip the language with a categorical denotational semantics which naturally relates to established models of quantum theory~\cite{dimeglio:rstar,caretteetal:sqrtpi}.
More precisely, we build on rig dagger categories with independent coproducts. This  semantics-first approach to language design fits in the larger programme of categorical quantum theory~\cite{heunenvicary:cqm,heunenkornell:hilb}. We employ these semantics to prove soundness of our compilation algorithm.

\subsubsection*{Implementation}
We include a prototype implementation of the combinator language~\cite{ricePhaseLanguageImplementation2025}.
This takes the form of a Rust executable and library, and can be built with \texttt{cargo} (tested with version 1.86.0). This implementation parses terms, performs typechecking, computes inverses and square roots, runs evaluation to a circuit, and finally outputs the matrix represented by the term. HTML documentation of this library is provided at \url{alexarice.github.io/phase-rs}.

\subsubsection*{Related work}

Our work is set in the general framework of quantum control
flow~\cite{valiron:habilitation,valiron:review}, as opposed to classical
control flow. In quantum control, the \emph{branching} is decided by quantum
data, leaving the result in a potential quantum superposition. There is a
number of ongoing research projects around quantum control flow, and more
specifically, on the integration of quantum control in programming languages
paradigms.

\begin{itemize}
	\item Symmetric pattern
		matching~\cite{sabryvalironvizzotto:control,phd-kostia,phd-louis} is a
		proof-of-concept programming language for quantum control, in which
		control is only quantum. Similar to a $\lambda$-calculus, its only
		primitives are type connectives, completed with complex numbers for the
		quantum aspect of the language; these complex numbers allow for the
		expression of any unitary operator in the language. However, like the
		$\lambda$-calculus, it is an abstract language, and does not reasonably
		compile to quantum circuits or any quantum hardware. Symmetric pattern
		matching can be seen as an improvement of
		QML~\cite{altenkirchgrattage:qml}, also entirely based on quantum if
		statements, but less scalable.
	\item There exists a whole research program on proof languages that include
		quantum control~\cite{diazcaro:review} based on intuitionistic linear
		logic. While this approach is fully scalable and mathematically sound,
		it is geared towards logical intuition and understandability rather than quantum programming.
	\item Most quantum programming languages in the literature are circuit
		description languages, possibly with a quantum if statement~\cite{ying:foundations,yingzhang:case,fuProtoQuipperReversingControl2024a}, effectively acting as classical programming languages whose
		values are quantum circuits. In these languages, unitaries come as
		constants that can be called and used as black boxes. This is hard to scale or automate.
	\item Qunity~\cite{qunity} mixes symmetric pattern matching and circuit
		description languages to allow for some form of quantum control while
		keeping a syntax relatively close to quantum circuits. It, however,
		still contains most unitary operations as constants. The compiler has exponential blow-ups so scalability is a challenge.
	\item Silq~\cite{bichseletal:silq,hirataheunen:qurts} has support for a quantum if but only
		for the type (qu)bit, and therefore lacks in scalability compared to
		what we are able to achieve. The language also does not come equipped
		with a compositional denotational semantics to support the validity of
		the operational semantics.
	\item The zeta calculus~\cite{botoforslund:zeta} is an abstract
		language---in the sense of the $\lambda$-calculus, which offers a
		compilation to the ZX-calculus, a graphical language for linear maps
		between finite-dimensional Hilbert spaces. However, the operations
		that the zeta calculus represent are not only unitary, since it allows
		to \emph{copy} and to discard on the Z and X bases. There is no known
		way of compiling the zeta calculus (or the ZX calculus) to quantum
		circuits.
	\item Universal quantum if conditional~\cite{bisioetal:qif}. On a more
		foundational aspect of quantum computing, it is known that there is no
		quantum operation that operates a generic ``quantum if'' on a black box oracle.  It means in
		particular that quantum theory does not allow for an operator $\lambda
		xy . \textsf{if}~x~\textsf{then}~y$.
\end{itemize}

\subsubsection*{Structure of this article}

After briefly reviewing the necessary background about quantum computing in \Cref{s:background}, we introduce the syntax of our combinator language in \Cref{s:syntax}, in addition to describing certain meta-operations on terms.
% \Cref{s:translations} then provides translations between these three versions.
Next, \Cref{s:examples} details four examples of important (families of) quantum algorithms in our language.
In \Cref{s:compiler}, we discuss a prototype compiler from the language into quantum circuits.
Denotational semantics are developed in \Cref{s:semantics}, and are used to exhibit equalities that hold within the language, and prove our compilation algorithm is sound.
We finally discuss different potential settings for our ``if let'' construction in \cref{sec:beyond-combinators}, describing alternative nominal representations of the core combinator language.
\Cref{s:conclusion} concludes by discussing future developments.%, and Appendices~\ref{app:typing}~and~\ref{app:semantics}  contain technical proofs relegated from the main text.

\subsubsection*{Acknowledgements}

% The material in Appendix~\ref{app:semantics} came out of discussions with the authors of~\cite{meglioetal:relations}.
The proofs of \Cref{th:ic-smc,th:rig} came out of discussions with the authors of~\cite{meglioetal:relations}.
We extend our thanks to the people of the Quantum Programming group in the
University of Edinburgh for their support and proofreading.  This research was
funded by the Engineering and Physical Sciences Research Council (EPSRC) under
project EP/X025551/1 "Rubber DUQ: Flexible Dynamic Universal Quantum
programming". C.M. is supported by the U.S. Army Research Oﬃce Grant No. W911NFF-23-1-0045 (Extensible and Modular Advanced Qubits). The views and conclusions contained herein are those of the authors and should not be interpreted as necessarily representing the official policies or endorsements, either expressed or implied, of the U.S. Government.

\section{Quantum Computing}\label{s:background}

We start with the essential background on quantum computing. We will only be concerned with unitary quantum computing, and have no need to consider the measurement readout at the end of the computation in detail. For more details we refer to textbooks such as~\cite{nielsenchuang,yanofskymannucci}.

\subsubsection*{Qubits}

The unit of quantum information is the \emph{qubit}. The state of a qubit is a unit vector $\varphi = \left(\begin{smallmatrix} x\\y \end{smallmatrix}\right)$ in the Hilbert space $\mathbb{C}^2$, that is, a pair of complex numbers $x$ and $y$ such that $|x|^2+|y|^2=1$. This vector is often written in \emph{ket notation} $\ket{\varphi}$. Two special vectors are the \emph{computational basis} states $\ket{0}=\left(\begin{smallmatrix} 1 \\ 0 \end{smallmatrix}\right)$ and $\ket{1}=\left(\begin{smallmatrix} 0 \\ 1 \end{smallmatrix}\right)$; a general state $\ket{\varphi}$ is in a \emph{superposition} of these two. Two such states that we will often use are $\ket{+}=\tfrac{1}{\sqrt{2}}\ket{0}+\tfrac{1}{\sqrt{2}}\ket{1}$ and $\ket{-}=\tfrac{1}{\sqrt{2}}\ket{0}-\tfrac{1}{\sqrt{2}}\ket{1}$.

\subsubsection*{Entanglement}

In a system with multiple qubits, the state is a unit vector in the \emph{tensor product}. For example, if the first qubit is in state $\ket{0}$, and the second qubit is in state $\ket{1}$, then the state of the compound system is the vector $\ket{0} \otimes \ket{1} \in \mathbb{C}^2 \otimes \mathbb{C}^2$, also written as $\ket{01}$. Not all states of a system with multiple qubits are of this form. For example, the state $\tfrac{1}{\sqrt{2}}\left(\ket{00}+\ket{11}\right)$ is \emph{entangled}: it cannot be written in the form $\ket{\varphi} \otimes \ket{\psi}$. In general, the states of a system with $n$ qubits can range over the unit vectors in $\mathbb{C}^2 \otimes \cdots \otimes \mathbb{C}^2 \simeq \mathbb{C}^{(2^n)}$.

\subsubsection*{Unitaries}

Qubits can undergo operations specified by unitary matrices. These are $2^n$-by-$2^n$ matrices $U$ with complex entries satisfying $U^\dag U = 1$ (and hence also $UU^\dag = 1$), that is, the matrix is invertible and its inverse is its conjugate transpose. On single qubits, standard operations include:
\[
  \Gate{H} = \frac{1}{\sqrt{2}}\begin{bmatrix} 1 & 1 \\ 1 & -1 \end{bmatrix}
  \qquad
  \Gate{X} = \begin{bmatrix} 0 & 1 \\ 1 & 0 \end{bmatrix}
  \qquad
  \Gate{Y} = \begin{bmatrix} 0 & -i \\ i & 0 \end{bmatrix}
  \qquad
  \Gate{Z} = \begin{bmatrix} 1 & 0 \\ 0 & -1 \end{bmatrix}
  \qquad
  \Gate{T} = \begin{bmatrix} 1 & 0 \\ 0 & e^{\nicefrac{i\pi}{4}} \end{bmatrix}
\]
The matrix $\Gate{H}$ is called the \emph{Hadamard} transformation, and satisfies $\Gate{H}\ket{0}=\ket{+}$ and $\Gate{H}\ket{1}=\ket{-}$.

\subsubsection*{Control}

Another way to combine two $n$-qubit systems is by direct sum as in the left-hand side of
\[
  \mathbb{C}^{(2^n)} \oplus \mathbb{C}^{(2^n)} \simeq \mathbb{C}^{(2^{n+1})} \simeq \mathbb{C}^2 \otimes \mathbb{C}^{(2^n)}\text.
\]
Notice how this `sum type' can also be described as a `product type' with a new qubit as in the right-hand side. This new qubit is called the \emph{control} qubit, because it controls which operation is applied to the other, \emph{target}, qubits, as follows.
Any $n$-qubit unitary $U$ can be extended to an $(n+1)$-qubit unitary $CU$, defined as \(I \oplus U\) or given by the block diagonal matrix
\[
CU = \begin{bmatrix}
I & 0 \\ 0 & U
\end{bmatrix}
\]
where $I$ is the identity on $\mathbb{C}^{(2^n)}$.
This controlled-$U$ will apply $U$ only if the control qubit was in the state $\ket{1}$; otherwise
it will do nothing. For example, the controlled-$X$ gate $CX$ is
given by
\[
\Gate{CX} = \left[\begin{smallmatrix}
  1 & 0 & 0 & 0 \\
  0 & 1 & 0 & 0 \\
  0 & 0 & 0 & 1 \\
  0 & 0 & 1 & 0
\end{smallmatrix}\right]
\]
Since the control qubit can be in superposition, both branches are executed in superposition, which is a key difference from classical control flow, where only one branch is executed.
Observe especially that this quantum control is predicated on the \emph{computational basis} $\{\ket{0},\ket{1}\}$ of the control qubit $\mathbb{C}^2$. Many important quantum algorithms need to perform controlled operations on different subspaces than those aligned along the linear spans of $\ket{0}$ and $\ket{1}$.

\subsubsection*{Phases}

Two unitary matrices $U$ and $V$ are indistinguishable in their measurable effect on qubits when they are equal up to a \emph{global phase}: $U=e^{i \theta}V$ for some $\theta \in [0,2\pi)$. In other words, we consider the identity matrix to model the same computation multiplying with the scalar $e^{i\theta}$.
Nevertheless, \emph{local phases}
\[
  \Gate{P}(\theta) = \begin{bmatrix}
    1 & 0 \\ 0 & e^{i\theta}
  \end{bmatrix}
\]
are \emph{not} identified with the identity, and in fact form the heart of many quantum algorithms, despite the fact that they can be regarded as controlled $0$-qubit global phase operations. For example, notice that $\Gate{Z}=\Gate{P}(\pi)$ and $\Gate{T}=\Gate{P}(\nicefrac{\pi}{4})$.

\subsubsection*{Circuits}

Any quantum computation is described by a unitary matrix. These are typically built up from one- and two-qubit unitaries, also called \emph{quantum gates}, that are combined using tensor products and matrix multiplication. It is customary to draw such a composite unitary as a \emph{quantum circuit}: a graphical depiction where horizontal wires represent qubits acted upon by quantum gates while they flow from left to right. A controlled $\Gate{U}$ gate is depicted as on the left below, and in particular the controlled $\Gate{X}$ gate is depicted as on the right below.
% \[\begin{quantikz}[row sep=.3cm]
%     &\ctrl{1}& \\
%     &\gate{\Gate{U}}&
%   \end{quantikz}
%   \qquad \qquad \begin{quantikz}[row sep=.3cm]
%     &\ctrl{1}& \\
%     &\targ{}&
%   \end{quantikz} \]
  \[
	  \tf{control}
  \]
Gates can similarly be controlled on multiple qubits, such as the $\Gate{CCX}$ matrix, or \emph{Toffoli gate}, in \cref{fig:toffoli}.
See \cref{fig:grover}  \todo[kim]{What is this figure showing?} for another example of a quantum circuit.

\subsubsection*{Eigendecomposition} For a linear map \(U : \mathcal{H} \to \mathcal{H}\), its \emph{eigenvalues} are scalars \(\lambda\) such that the subspace \(\{v : Uv = \lambda v\}\) is non-zero. Such subspaces are referred to as \emph{eigenspaces} and the vectors they contain are called \emph{eigenvectors} of the map \(U\). When \(U\) is unitary, there exists a basis \(\{\ket {\lambda_i}\}\) of eigenvectors for \(\mathcal{H}\), each with eigenvalue \(\lambda_i\), allowing us to obtain the \emph{eigendecomposition}:
\[ U = \sum_i \lambda_i \ketbra{\lambda_i}\]
The maps \(\ketbra{\lambda_i}\) are \emph{projections}, maps \(p\) such that \(p^2 = p\).

Further, \(U\) admits a \emph{diagonalisation} \(U = Q\Lambda Q^\dagger\), where \(\Lambda\) is the diagonal matrix with entries \(\lambda_i\). In this sense, every unitary map can be realised as a sequence of phase rotations applied to its eigenspaces, motivating our representation of them utilising \emph{just a phase}.

\section{Syntax}\label{s:syntax}

We are now positioned to introduce the core language of this paper, a combinator-based language for describing unitary linear transformations, with two basic building blocks: a global phase unitary, and a quantum ``if let'' allowing a restricted form of pattern matching.

In this combinator-style language, terms represent unitary maps on a set number of qubits. The types for unitaries are therefore very simple, and are in direct correspondence with natural numbers: a term \(t\) of type \(\unitary{n}\) will represent a unitary map \(\C^{2^n} \to \C^{2^n}\), and this is the only possible type a unitary can have. This language has no variables, and hence its typing derivations do not require a context and are simply written:
\[\vdash t : \unitary{n}\]
for a term \(t\) and \(n \in \mathbb{N}\).

We introduce a (global) phase operation as the only primitive ``gate''. It takes the form of a \(0\)-qubit unitary (and should not be confused with the 1-qubit phase gate commonly referred to as \(\Gate{S}\)). For each angle\footnote{In practice we must fix a (countable) group of angles in order for this syntax to be finitary.} \(\theta \in \mathbb{R}\) we write:
\begin{mathpar}
  \inferrule{ }{\vdash \Phase{\theta} : \unitary{0}}
\end{mathpar}

To create larger programs, we must be able to compose unitaries together. This can be done sequentially or in parallel. We further require an explicit identity term. These have the following typing rules:
\begin{mathpar}
  \inferrule{\vdash s : \unitary{n} \and \vdash t : \unitary{n}}{\vdash s; t : \unitary{n}}\and
  \inferrule{\vdash s : \unitary{n} \and \vdash t : \unitary{m}}{\vdash s \otimes t : \unitary{(n + m)}} \and
  \inferrule{ }{\vdash \id{n} : \unitary{n}}
\end{mathpar}\todo[malin]{Should id have a subscript when qn doesn't?}

At this point, this syntax can only represent unitaries that perform a global phase, which, as already noted in \cref{s:background}, have no computational effect in a quantum circuit. The ability to perform arbitrary quantum gates will be derived from our quantum ``if let'' construction. This construction allows a unitary to be performed on a subspace specified by a \emph{pattern}. Patterns $p$ are given types of the form \(\pattern{n}{m}\) and correspond to isometries \(i \colon \C^{2^n} \to \C^{2^m}\). The ``if let'' expression then has the following typing rule:
\begin{mathpar}
  \inferrule{\vdash p : \pattern{n}{m} \and \vdash s : \unitary{n}}{\vdash \ifthen{p}{s} : \unitary{m}}
\end{mathpar}
One intuition for the action of the ``if let'' expression is the following: if \(s\) represents the unitary \(U\), and \(p\) represents the isometry \(i\), then the unitary represented by ``\(\ifthen p s\)'' performs $U$ on the range subspace of $i$, and the identity on its orthogonal complement. The ``if let'' construction can be viewed as a restricted form of symmetric pattern matching~\cite{sabryvalironvizzotto:control}.

Patterns are given by a separate but related syntax to terms, for which the rules are given below:
\begin{mathpar}
  \inferrule{ }{\vdash \ket{0} : \pattern{0}{1}}\and
  \inferrule{ }{\vdash \ket{1} : \pattern{0}{1}}\and
  \inferrule{ }{\vdash \ket{+} : \pattern{0}{1}}\and
  \inferrule{ }{\vdash \ket{-} : \pattern{0}{1}}\and
  \inferrule{\vdash s : \unitary{n}}{\vdash s : \pattern{n}{n}}\and
  \inferrule{\vdash p : \pattern{n}{m} \and \vdash q : \pattern{l}{n}}{\vdash p \cdot q : \pattern{l}{m}}\and
  \inferrule{\vdash p : \pattern{n}{m} \and \vdash q : \pattern{n'}{m'}}{\vdash p \otimes q : \pattern{(n + n')}{(m + m')}}\and
\end{mathpar}
The latter three rules observe that all unitary maps are also isometries, and that isometries are closed under composition and tensor products. Each pattern \(\ket x\) represents the isometry \(z \mapsto z\ket x : \C \to \C^2\). We note that the composition for patterns is in function composition order, in contrast to the diagrammatic composition order for terms.

We highlight three important cases of this construction:
\begin{itemize}
\item If \(p : \pattern{0}{n}\), then the term \(\ifthen{p}{\Phase{\theta}}\) represents the unitary which maps \(p(\alpha) + v\) to \(e^{i\theta}p(\alpha) + v\) (where \(\langle p(1), v\rangle = 0\)), which has eigenvalues \(1\) and \(e^{i\theta}\).
\item Let \(s : \unitary{n}\) represent the unitary \(U\) and consider the term:
  \[ \ifthen {\ket 1 \otimes \id{n}} s \]
The unitary represented by this term sends any input of the form \(\ket 1 \otimes v\) to \(\ket 1 \otimes U(v)\), and leaves any input of the form \(\ket 0 \otimes v\) unchanged. This term therefore represents the controlled \(U\) operation.
\item Suppose \(s, t : \unitary{n}\), representing unitaries \(U\) and \(V\). Then the unitary represented by
  \[\ifthen s t\]
    is the unitary \(U \circ V \circ U^\dagger\), which sends \(U(v)\) to \(U(V(v))\).
\end{itemize}

\begin{example}[\(\Gate{X}\) gate]
  Our first example is the term:
  \[ \ifthen {\ket -} {\Phase{\pi}} \]
  By the intuition above, this represents a unitary which maps \(\ket -\) to \(e^{i\pi}\ket - = - \ket -\), and \(\ket +\) (which is orthogonal to \(\ket -\)) to \(\ket +\). Its action on other vectors is determined by linearity; \(\ket 0 = \frac{1}{\sqrt 2}\left(\ket + + \ket -\right)\) is sent to \(\frac{1}{\sqrt 2}\left(\ket + - \ket -\right) = \ket 1\) and similarly \(\ket 1\) is sent to \(\ket 0\), making this the quantum \(\Gate{X}\) gate.
\end{example}

The syntax presented here allows the simple definition of two important meta-level operations: inversion and exponentiation. The ability to obtain the inverse of a quantum program is not uncommon, yet we highlight the simplicity of the definition below.
\begin{definition}[Inversion]\label{def:syn-inverse}
  Given a term \(\vdash t : \unitary{n}\), we define its \emph{inverse} \(\vdash t^\dagger : \unitary{n}\) by structural induction on the syntax:
  \[ \Phase{\theta}^\dagger = \Phase{-\theta} \qquad (\ifthen p s)^\dagger = \ifthen p {s^\dagger} \qquad (s \otimes t)^\dagger = s^\dagger \otimes t^\dagger \qquad (s ; t)^\dagger = t^\dagger;s^\dagger\]
  A simple induction shows the resulting term is well-typed.
\end{definition}

For terms which do not contain the composition constructor (in particular terms which consist of a single ``if let'' statement), we can perform the much more general operation of exponentiation. The ability to define exponentiation exemplifies the utility of our syntax.

\begin{definition}[Exponentials]
  Let \(\vdash t : \unitary{n}\) be a ``composition-free'' term, i.e. a term containing no instances of ``\(;\)''. For a real number \(\alpha\), define the \emph{exponentiation} \(\vdash t^\alpha : \unitary{n}\) by structural induction:
  \[ \Phase{\theta}^\alpha = \Phase{\alpha \theta} \qquad (\ifthen p s)^\alpha = \ifthen p {s^\alpha} \qquad (s \otimes t)^\alpha = s^\alpha \otimes t^\alpha\]
  Similar to above, a simple induction shows exponentiation is well-typed. We note that the case where \(\alpha = -1\) coincides with the inversion operation. We write \(\sqrt t\) for \(t^{\nicefrac{1}{2}}\).
\end{definition}
We highlight the use of the exponentiation operation by applying it to the \(\Gate{X}\) gate.
\begin{example}
  The quantum \(\Gate{V} = \sqrt{\Gate{X}}\) gate, which satisfies \(\Gate V \circ \Gate V = \Gate X\) is given by the matrix:
  \[ \Gate{V} = \frac{1}{2}
    \begin{bmatrix}
      1 + i & 1 - i\\
      1 - i & 1 + i
    \end{bmatrix}
  \]
  Deriving this matrix from only the definition of \(\Gate{X}\) (or indeed obtaining a quantum circuit for this gate) is non-trivial, yet its definition in our language is immediate from the definition of the \(\Gate{X}\) and exponentiation:
  \[ \Gate{V} = \Gate{X}^{0.5} = \ifthen {\ket -} {\Phase{\nicefrac{\pi}{2}}}\]
  This term represents the unitary given by the matrix above.
\end{example}

We are now able to recover the definitions of many common quantum gates, which are given in \cref{fig:defs}. We emphasise that these gates are approximately universal, and hence all unitaries (of dimension \(2^n\)) can be represented using this language.
\begin{figure}
  \centering
  \begin{equationarray*}{rrlcrrll}
    \Gate{Z} &\coloneq& \ifthen{\ket{1}}{\Phase{\pi}} & \qquad & \Gate{S} &\coloneq & \sqrt{\Gate{Z}} &= \ifthen{\ket{1}}{\Phase{\nicefrac{\pi}{2}}}\\
    \Gate{X} &\coloneq& \ifthen{\ket{-}}{\Phase{\pi}} & \qquad & \Gate{V} &\coloneq & \sqrt{\Gate{X}} &= \ifthen{\ket{-}}{\Phase{\nicefrac{\pi}{2}}}\\
    \Gate{Y} &\coloneq& \ifthen{\ket{-}}{\Phase{\pi}} & \qquad & \Gate{T} &\coloneq & \sqrt{\Gate{S}} &= \ifthen{\ket{1}}{\Phase{\nicefrac{\pi}{4}}}\\
    \Gate{CZ} &\coloneq& \ifthen{\ket{1} \otimes \ket{1}}{\Phase{\pi}} & \qquad\qquad\qquad & \Gate{CX} &\coloneq & \multicolumn{2}{l}{\ifthen{\ket{1} \otimes \ket{-}}{\Phase{\pi}}}\\
    \Gate{H} &\coloneq&\multicolumn{6}{l}{\ifthen{\Gate{Y}^{\nicefrac{1}{4}} \cdot \ket{1}}{\Phase{\pi}}}\\
    &=&\multicolumn{6}{l}{\ifthen {(\ifthen {S \cdot \ket{-}} {\Phase{\nicefrac{\pi}{4}}}) \cdot \ket 1} {\Phase{\pi}}}
  \end{equationarray*}
  \caption{Definitions of common quantum gates in the combinator syntax. A fully universal gate set can be generated with \emph{just a phase}.}
  \label{fig:defs}
\end{figure}

\begin{example}
  The 5-qubit GHZ state (\(= \nicefrac{1}{\sqrt{2}}(\ket {00000} + \ket {11111})\)) can be prepared from the zero state as follows, where \(\Gate{H}\) and \(\Gate{X}\) are defined in \cref{fig:defs}:
  \[\Gate{H} \otimes \id{4}; \ifthen{\ket{1}\otimes \id{4}}{\Gate{X} \otimes \Gate{X} \otimes \Gate{X} \otimes \Gate{X}}\]
  This highlights one of the flaws of the combinator style syntax, to apply a Hadamard gate to the first qubit, we must explicitly tensor it with the remaining qubits. Further, the pattern \(\ket{1} \otimes \id{4}\) causes the body to be controlled by the first qubit, but also must explicitly tensored with the remaining qubits.

  The combinator syntax also enforces a total ordering on the qubits, which may or may not be desirable. If we had decided to create the GHZ state with successive \(\Gate{CX}\) gates, we would find that there is no trivial way to apply such a 2-qubit gate to the first and third qubits.
\end{example}

\begin{remark}
  Instead of introducing \(\ket +\) and \(\ket -\) as primitive patterns, we could have instead introduced the Hadamard gate \(\Gate{H}\) as a primitive unitary, defining \(\ket + = \Gate{H} \cdot \ket 0\) and \(\ket - = \Gate{H} \cdot \ket 1\). Presenting the language in this way may be beneficial for contexts where the Hadamard gate is an important or easy operation, as the definition of the Hadamard gate in \cref{fig:defs} is more involved. The set up taken above, however, allows an arguably more minimal presentation by having the phase rotation be the only primitive unitary, and enables the exponentiation operation; it is unclear what the square root of a primitive Hadamard operation should be, yet its definition is immediate when presented as a single ``if let'' statement.
\end{remark}

We end this section with one further illustrative example.

\begin{example}
  Let \(\Gate{XC} = \ifthen {\ket{-} \otimes \ket{1}} {\Phase{\pi}}\), a controlled not operation where the second qubit is the control qubit and the first is the target. We can then define concisely the swap gate as:
  \[\Gate{Swap} \coloneq \ifthen {\Gate{CX}} {\Gate{XC}} \]
  The unitary pattern \(\Gate{CX}\) acts on the body of the ``if let'' by conjugation, allowing us to recover the more usual definition \(\Gate{CX}; \Gate{XC}; \Gate{CX}\) of the swap. By a simple manipulation, we can also recover another definition of the swap gate:
  \[\ifthen {\Gate{CX} \cdot (\ket{-} \otimes \ket{1})} {\Phase{\pi}}\]
  This presents an alternative way of understanding the action of this gate. By observing that \(\Gate{CX} \ket{-1} = \frac{1}{\sqrt 2}(\ket{01} - \ket{10})\), we notice that the term above multiplies this component of the input by \(-1\), permuting the \(\ket{01}\) and \(\ket{10}\) components.
\end{example}

% \section{Translations}\label{s:translations}

\section{Algorithms}\label{s:examples}

Many well-known quantum algorithms can be expressed naturally in terms of conditional phases.
Below, we give implementations of several of the most prominent algorithms and prove their correctness.
We begin with Grover's search algorithm, whose \emph{oracle} and \emph{diffusion} operators are precisely conditional phases on the target subspace and the uniform subspace, respectively.
Next, quantum simulation algorithms are often based on finite-difference time-domain Hamiltonian evolution. In these algorithms, we take a number of time-steps, in each of which the state advances by applying conditional phases corresponding to the eigenspaces of the Hamiltonian.
After that, we show how the \emph{quantum Fourier transform} (QFT) is naturally expressed as a sequence of conditional phases.
Finally, we present implementations of \emph{quantum signal processing} and the \emph{quantum eigenvalue transform}, which implement functions of black-box unitaries by applying conditional phases.

\subsection{Grover's Algorithm}

The celebrated quantum database search algorithm of Grover~\cite{grover} has a simple formulation in this language.
Assume we are searching a database $X = \{0, 1, \dots, N - 1\}$ of size $N = 2^{n}$ for elements $x$ on which a function $f : X \rightarrow \{0, 1\}$ takes the value $1$. We may assume that $f(x) = 1$ at only a single element $0 \le \omega < N$, as multiple marked elements can be handled by a sequence of such oracles.
Recall that the algorithm consists of three steps~\cite{nielsenchuang}:

\begin{enumerate}
  \item Preparation of a uniform superposition $\ket{s} = \frac{1}{\sqrt{|X|}} \sum_{x \in X} \ket{x}$

  \item Repeat $\lceil \pi \sqrt{N} / 4 \rceil$ times:
        \begin{enumerate}
          \item Apply the \emph{oracle operator} $U_{f} = \sum_{x \in X} (-1)^{f(x)} \ketbra{x} = 1 - 2 \ketbra{\omega}$.
                \item Apply the \emph{diffusion operator} $U_{s} = 2 \ketbra{s} - I$.
        \end{enumerate}
  \item Measure the quantum state.
\end{enumerate}
With high probability, the measurement result is $\ket{\omega}$.

The oracle operator $U_{f}$ can be implemented by the following program:
\[ \ifthen {\ket{\omega_{0}} \otimes \cdots \otimes \ket{\omega_{n - 1}}} {\Phase{\pi}}\]
  where $\omega_{j}$ is the $j$\textsuperscript{th} bit in the binary expansion of $\omega$. Similarly, the diffusion operator \(U_s\) is given by the program
$
    \Phase{\pi} \otimes \id n ; \ifthen{\ket + \otimes \cdots \otimes \ket +} {\Phase \pi}
$.
Compare these pieces of syntax with the circuits in \Cref{fig:grover}.
Grover's algorithm is then simply an interleaving of these operators.

\subsection{Quantum Simulation}

We will now implement the Trotter simulation algorithm~\cite{lloyd1996universal} by applying a sequence of conditional phases.
This is perhaps not how computer scientists naturally think about quantum simulation, but physicist practitioners naturally think in terms of dynamics as composing programs `spectrally' and applying differential phases.
This is completely independent from the mechanism of the if-let construct, which we focus on.
Once the Hamiltonian is decomposed into projections, we construct a program that realises conditional phases on the subspace picked out by each projection.

Let $H \in \mathcal{B}(\mathcal{H})$ be a positive self-adjoint operator on a $2^{n}$-dimensional Hilbert space.
Then there is a (possibly empty) decomposition $H = \sum_{i = 1}^{K} \lambda_{i} \Pi_{i}$, where each $\Pi_{i}$ is a projection. Write $\tilde{H} = ((\lambda_{i}, \Pi_{i}))_{i = 1}^{K}$ for the $K$-tuple of spectral components consisting of ordered pairs of eigenvalues and projectors. Now, the decomposition above is not unique, so neither is $\tilde{H}$. In fact, it need not be a \emph{spectral} decomposition in the usual sense. We have two choices: impose uniqueness by requiring that the $\Pi_{i}$ be orthonormal, that the range of $\Pi_{i}$ coincides with the kernel of $H - \lambda_{i} I$, that $\lambda_{1} \le \dots \le \lambda_{K}$; or, take $\tilde{H}$ to be primary and $H$ to be derived. Let us adopt the latter approach.

Now suppose we are given a set of patterns $\{p_{i}\}_{1 \le i \le K}$.
By possibly padding out the range of the projection $\Pi_i$ to power of two dimension, we may assume without loss of generality that $\vdash p_i : \pattern{m_i}{n}$, that is \(\Pi_i = \iota_i \iota_i^\dagger\) where \(\iota\) is the isometry represented by \(p_i\).

Then we define a program $U_{\tilde{H}}(t)$ inductively, by
\begin{align*}
  U_{\emptyset}(t) &\coloneq \id{n} \\
  U_{((\lambda_{i}, \Pi_{i}))_{i = 1}^{k + 1}}(t) &\coloneq \ifthen{p_{k + 1}}{\Phase{-\lambda_{i + 1} t} \otimes \id{m_i}}; U_{((\lambda_{i}, \Pi_{i}))_{i = 1}^{k}} (t).
\end{align*}
If the $\Pi_{i}$ are mutually orthogonal, then $U_{\tilde{H}}(t)$ represents the unitary \(e^{- i H t}\). This does not hold in general, for non-commuting projectors, but \(U_{\tilde H}(t/N)^N\) represents a unitary which converges to \(e^{-iHt}\) as \(N \to \infty\) by the well-known \emph{Trotterisation} formula~\cite{trotter1959product}.

% \subsubsection{Example: Interacting dipoles}
Let us now consider a simple concrete example of interacting spin-1/2 magnetic dipoles in an external magnetic field $\boldsymbol{B}$~\cite{sakurai2020modern}. The Hamiltonian is
\begin{equation*}
    H = H_{\mathrm{free}} + H_{\mathrm{int}} \qquad H_{\mathrm{free}} = - \frac{\hbar}{2} \sum_{j = 1}^{2} \gamma_{j} \boldsymbol{\sigma}_{i} \cdot \boldsymbol{B} \qquad H_{\mathrm{int}} = \frac{\mu_{0} \gamma_{1} \gamma_{2} \hbar^{2}}{16 \pi r^{3}} \left( \boldsymbol{\sigma}_{1}\cdot\boldsymbol{\sigma}_{2} - 3 (\hat{\boldsymbol{r}} \cdot \boldsymbol{\sigma}_{1})(\hat{\boldsymbol{r}} \cdot \boldsymbol{\sigma}_{2}) \right)
\end{equation*}
where $\boldsymbol{\sigma}_{j}$ is the vector of Pauli operators on the $j^{\mathrm{th}}$ spin, $\gamma_{j}$ is the gyromagnetic moment of the $j^{\mathrm{th}}$ spin, $\boldsymbol{r}$ is the displacement between the two spins, $r = \lVert \boldsymbol{r} \rVert$, and $\hat{\boldsymbol{r}} = \boldsymbol{r} / r$. We can simplify this expression to
$
  H
  %&= \omega_{1} \sigma^{z}_{1} + \omega_{2} \sigma^{z}_{2} + J \left(\sigma^{x}_{1} \sigma^{x}_{2} + \sigma^{y}_{1} \sigma^{y}_{2} - 2 \sigma^{z}_{1} \sigma^{z}_{2}\right) \\
  = \omega_{1} \sigma^{z} \otimes I + \omega_{2} I \otimes \sigma^{z} + J \left(\sigma^{x} \otimes \sigma^{x} + \sigma^{y} \otimes \sigma^{y} - 2 \sigma^{z} \otimes \sigma^{z}\right)$.
% \end{align*}
Not having imposed uniqueness or even orthogonality on the decomposition $\tilde{H}$, we can work term-by-term. Writing
\[\Pi_{+z} = \ketbra{0} \quad \Pi_{-z} = \ketbra{1} \quad \Pi_{+x} = \ketbra{+} \quad \Pi_{-x} = \ketbra{-} \quad \Pi_{+y} = \ketbra{i} \quad \Pi_{-y} = \ketbra{-i}\text, \]
we obtain
\begin{alignat*}{5}
  H ={} &\omega_{1} (\Pi_{+z} - \Pi_{-z}) &&\otimes (\Pi_{+z} + \Pi_{-z}) &&{}+{} &&\omega_{2} (\Pi_{+z} + \Pi_{-z}) &&\otimes (\Pi_{+z} - \Pi_{-z}) \\
   {}+{} &J (\Pi_{+x} - \Pi_{-x}) &&\otimes (\Pi_{+x} - \Pi_{-x}) &&   {}+{} &&J (\Pi_{+y} - \Pi_{-y}) &&\otimes (\Pi_{+y} - \Pi_{-y}) \\
   {}-{} &2J (\Pi_{+z} - \Pi_{-z}) &&\otimes (\Pi_{+z} - \Pi_{-z})\text.
\end{alignat*}
Finally, we can distribute the tensor products over the sums, obtaining the decomposition $\tilde{H}$ into projectors. It remains to find the corresponding \textit{patterns}. Recall that for each $\Pi_{i}$ we must find a corresponding pattern $p_{i}$. The following patterns $p_{\pm x}, p_{\pm y}, p_{\pm z}$ suffice:
\begin{equation*}
  p_{+z} = \ket{0} \quad p_{-z} = \ket{1} \quad p_{+x} = \ket + \quad p_{-x} = \ket - \quad p_{+y} = \Gate{S} \cdot \ket + \quad p_{-y} = \Gate{S} \cdot \ket -
\end{equation*}
These patterns allow us to compute a term representing \(e^{-iHt}\), as required.

\subsection{Quantum Fourier Transform}

The quantum Fourier transform over $\mathbb{Z} / 2^{n} \mathbb{Z}$ is the $n$-qubit unitary operator \(F_{2^n}\) such that $\bra{y} F_{2^{n}} \ket{x} = 2^{- n / 2} \omega^{xy}$, for $0 \le x, y < 2^{n}$, where $\omega$ is a primitive $2^{n}$\textsuperscript{th} root of unity~\cite{nielsenchuang}. If $x_{1}x_{2}\dots x_{n}$ is the binary expansion of $0 \le x < 2^{n}$, then
\begin{equation*}
    F_{2^{n}} \ket{x_{1}} \otimes \cdots \otimes \ket{x_{n}} = 2^{- n / 2} \bigotimes_{j = 1}^{n} \left(\ket{0} + e^{2 \pi i x 2^{j - 1 - n}}\ket{1}\right).
\end{equation*}
The textbook algorithm for implementing this operation consists of a sequence of controlled phases. Define the \emph{dyadic rational phase gate}
$
  \Gate{R}_n \coloneq \ifthen{\ket{1}}{\Phase{\nicefrac{2\pi}{2^n}}}
$.
To define \(\vdash \QFT{n} : \unitary{n}\), the term representing the \(n\) qubit quantum Fourier transform, we can leverage that chains of controlled gates with the same control qubit can be replaced by a single control block. Using this we obtain the following recursive definition:
\begin{align*}
  \QFT{0} &\coloneq \id{0}\\
  \QFT{n+1} &\coloneq \Gate{H} \otimes \id{n}; \ifthen {\ket 1 \otimes \id{n}} {\Gate{R}_2 \otimes \cdots \otimes \Gate{R}_{n+1}} ; \id{1} \otimes \QFT{n}
\end{align*}
In order recover the original quantum Fourier transform \(F_{2^n}\), the order of the output qubits must be inverted, which could be done with the addition of the appropriate \(\Gate{Swap}\) gates.

\subsection{Quantum Signal Processing}\label{sec:qsp}

\emph{Quantum signal processing} (QSP) is a procedure that transforms a parametric single-qubit rotation to modify its sensitivity to the parameter~\cite{low2016methodology,low2017optimal}. Following~\cite{martyn2021}, we write
\begin{equation*}
    W(a) = e^{i \theta X} = \begin{bmatrix}
      a & i \sqrt{1 - a^{2}} \\
      i \sqrt{1 - a^{2}} & a
      \end{bmatrix}
\end{equation*}
for the parametric unitary to which the user has black-box access (a rotation about the $X$-axis by $\theta = -2 \cos^{-1} a$). QSP makes a number of calls to $W(a)$, as well as a number of phases, to implement a modified unitary,
\begin{equation*}
    \tilde{W}(a) = \begin{bmatrix}
      P(a) & i Q(a) \sqrt{1 - a^{2}} \\
      i Q^{*}(a) \sqrt{1 - a^{2}} & P^{*}(a)
      \end{bmatrix}
\end{equation*}
This is useful in, among other applications, quantum control of large ensembles of quantum systems subject to inhomogeneous coherent control~\cite{brown2004arbitrarily}.
We regard $a$ as the amplitude of a control field, which due to spatial gradients couples more strongly to some elements of an ensemble than to others.
When the amplitude of the control field varies over an interval $I = [a_{0}, a_{1}]$ within a sample, we may choose $P, Q$ such that $\tilde{W}(I)$ is approximately constant.
Dually, in sensing applications, we may wish to \emph{enhance} sensitivity to the parameter $a$.

There turns out to be a construction of a composite control sequence that implements any $\tilde{W}$ for any polynomials $P$ (resp. $Q$) of degree $d$ (resp. $d - 1$) and parity $d \bmod 2$ (resp. $(d - 1)\bmod 2$), subject to the constraints imposed by unitarity of $\tilde{W}$, using $d$ calls to $W$.
All such $\tilde{W}$ can be realised as a product,
\begin{equation*}
    \tilde{W}(a) = W_{\vec{\phi}}(a) \coloneqq S(\phi_{0}) W(a) S(\phi_{1}) \cdots S(\phi_{d - 1}) W(a) S(\phi_{d}),
\end{equation*}
for some tuple $\vec{\phi} \in [0, 2 \pi)^{d + 1}$, where $S(\phi) = e^{i \phi_{0} Z}$.

As a single-qubit protocol, the translation to a program is easy. Define
\begin{align*}
  \Gate{R}_{z}(\alpha) &\coloneq \Phase{- \alpha / 2} \otimes \id{1}; \ifthen{\ket{0}}{\Phase{\alpha}} \\
  \Gate{R}_{x}(\alpha) &\coloneq \Phase{- \alpha / 2} \otimes \id{1}; \ifthen{\ket{+}}{\Phase{\alpha}}.
\end{align*}
We then define a \emph{quantum signal processing} program $\mathsf{QSP}(a; \vec{\phi})$ inductively by
\begin{align*}
  \mathsf{QSP}(a; (\phi_{0})) &\coloneq \Gate{R}_{z}(2 \phi_{0}) \\
  \mathsf{QSP}(a; (\phi_{0}, \dots, \phi_{k}, \phi_{k + 1})) &\coloneq \mathsf{QSP}(a; (\phi_{0}, \dots, \phi_{k})); \nonumber \\
  &\hspace{2em} \Gate{R}_{x}(-2 \cos^{-1}(a)); \Gate{R}_{z}(2 \phi_{k + 1})
\end{align*}
(recalling that $\vec{\phi}$ is by hypothesis of length $\ge 1$).

\subsection{Quantum Eigenvalue Transform}

QSP generalises rather dramatically to the \emph{quantum eigenvalue transform} (QET)~\cite{martyn2021,low2024quantum}.
Like QSP, the QET applies a polynomial transform to an operator, but in this instance an operator on a larger finite-dimensional quantum system, not merely a qubit.
Still following the presentation and conventions of~\cite{martyn2021}, we proceed to an implementation of the QET.

We are given a Hamiltonian $H$ acting on a finite-dimensional Hilbert space $\mathcal{H} = \mathbb{C}^{2^m}$, and unitary $U$ acting on $\mathbb{C}^{2^n}$ such that $U$ contains a copy of $H$ in a block determined by a projector \(\Pi\), which is also given to us. Then the QET produces a unitary $\tilde U$, such that (using the notation of \cite{martyn2021}):
\begin{equation*}
  U = \kbordermatrix{&\Pi & \\
          \ \Pi\!\!\!\!& \mathcal{H}     & \cdot \\
          & \cdot & \cdot
        }\qquad
\tilde U = \kbordermatrix{&\Pi & \\
          \ \Pi\!\!\!\!& P(\mathcal{H})     & \cdot \\
          & \cdot & \cdot
        }
\end{equation*}
where $P(x)$ is some degree \(d\) polynomial function.

As an example, in the case where \(\Pi = \ketbra 0\), we may take:
\begin{equation*}
    U = \Gate{Z} \otimes H + \Gate{X} \otimes \sqrt{I - H^{2}} = \sum_{\lambda \in \sigma(H)} R(\lambda) \otimes \ketbra{\lambda}
\end{equation*}
where $\Gate{R}(\lambda)$ is the operator \(\lambda \Gate{Z} + \sqrt{1 - \lambda^2}\Gate{X}\) qubit operator and $\lambda$ ranges over the eigenvalues of $H$ such that \(H = \sum_\lambda \lambda \ketbra{\lambda}\).
% In particular, from~\cref{eq:qet-unitary} one obtains $R(\lambda) = \lambda Z + \sqrt{1 - \lambda^{2}}X$.

% In each qubit subspace corresponding to an eigenvalue $\lambda$, we may play the same game as before on the auxiliary qubit, in an eigenvalue-wise sense. This time, the $Z$-phase operators $S(\phi)$ will be conditional, since the phases are $\lambda$-dependent. To this end we define,
% \begin{equation}
%     \Pi_{\phi} \coloneqq \sum_{\lambda} e^{i \phi Z} \otimes \ketbra{\lambda}.
% \end{equation}
% Now we wish to implement this unitary.
% Martyn \emph{et al.} give an construction using projector-controlled-NOT gates sandwiching an unconditional $Z$-phase.
% However, I'm not entirely sure what they intend by this---we have already that $\Pi_{\phi} = S(\phi) \otimes I_{\mathcal{H}}$ algebraically, so it can be realised simply by applying $S(\phi)$ on the target qubit.
% Perhaps in their figure 3, the $\ket{0}$ wire is not the auxiliary qubit, but an \emph{extra} ancilla, while the $\ket{\psi_{0}}$ wires are the whole $\mathbb{C}^{2} \otimes \mathcal{H}$ composite system; this could make sense if in fact we have only a total Hilbert space $\overline{\mathcal{H}} \simeq \mathbb{C}^{2} \otimes \mathcal{H}$, rather than on-the-nose equality.
% But up to this point, they indicate no such thing!
% Only in Theorem 3 does it look (``with the location of $H$ determined by projector $\Pi$'') like this might be their intention.
% Until now, $\Pi = \ketbra{0}$.
% I will go forward with this assumption, which makes the rest of the algorithm just about as simple as QST
The crux of QET is that there is a tuple \(\vec{\phi} = (\phi_1, \dots, \phi_d)\) such that
\begin{equation*}
    \tilde{U} \coloneq \begin{cases}
      \prod_{k = 1}^{d / 2} \Pi_{\phi_{2 k - 1}} U^{\dagger} \Pi_{\phi_{2k}} U & d~\text{even} \\
      \Pi_{\phi_{1}} U \prod_{k = 1}^{(d - 1) / 2} \Pi_{\phi_{2 k}} U^{\dagger} \Pi_{\phi_{2k + 1}} U & d~\text{odd}
    \end{cases}
\end{equation*}
where \(\Pi_{\phi}\) is the ``projector controlled phase shift'' and is defined to be:
\[ \Pi_{\phi} = e^{i\phi(2\Pi - I)}\]
If we assume that the projector \(\Pi\) is provided to us via a pattern \(\vdash p_\Pi : \pattern{m} {n}\) (such that \(p_\Pi\) represents an isometry \(\iota\) such that \(\Pi = \iota\iota^\dagger\)), we can implement this projector controlled phase shift as a single ``if let'':
\[ \Gate{R}_{p_\Pi}(\phi) \coloneq \Phase{-\phi} \otimes \id{n}; \ifthen {p_\Pi} {\Phase{2\phi} \otimes \id{m}} \]
If we further assume we are given a term \(\vdash s_{U} : \unitary{n}\) which represents the unitary \(U\), then we define \(\mathsf{QET}(s_U, p_\Pi, \vec{\phi})\) inductively by:
\begin{align*}
  \mathsf{QET}(s_{U}, p_\Pi; ()) &\coloneq \id{n} \\
  \mathsf{QET}(s_{U}, p_\Pi; (\phi_{1})) &\coloneq s_{U}; \mathsf{R}_{p_\Pi}(\phi_{1}) \\
  \mathsf{QET}(s_{U}, p_\Pi; (\phi_{1}, \dots, \phi_{k - 1}, \phi_{k})) &\coloneq \mathsf{QET}(s_{U}; (\phi_{0}, \dots, \phi_{k - 2})); \\
  &\hspace{2em} s_{U}; \mathsf{R}_{p_\Pi}(\phi_{k}); s_{U}^\dagger; \mathsf{R}_{p_\Pi}(\phi_{k - 1})
\end{align*}
Note that there are two base cases: one each for $d$ even and $d$ odd.

% \begin{proposition}
% \begin{equation}
%     \llbracket \mathsf{QET}(s_{U}; \vec{\phi}) \rrbracket = U_{\vec{\phi}}.
% \end{equation}
% \end{proposition}
% \begin{proof}
%   By a straightforward induction on $d$.
% \end{proof}

% A limitation of this implementation is that we simply assume
% \begin{equation}
% (\Pi \otimes I_{\mathcal{H}}) \sem{s_{U}} (\Pi \otimes I_{\mathcal{H}}) = (\Pi \otimes I_{\mathcal{H}}) H (\Pi \otimes I_{\mathcal{H}}).
% \end{equation}
% There is no general prescription to realise such a subroutine for arbitrary $H$ (within some reasonably-expressive class).
% The operator $H$ is simply not a part of the ontology of this program.
% We must make the same assumption about its inverse; in principle, the combinator language could be extended to include daggers, to eliminate that extra assumption.

% \subsection{QSVT}
% \TODO{Continuation from QET} this is now the generalization of above to the case where $H$ is not necessarily square.

\section{Compilation}\label{s:compiler}

In this section we will describe a conversion from the combinator language in \cref{s:syntax} to a more traditional quantum circuit representation of unitary programs.
Our strategy for the compilation will be to provide a normalisation algorithm, putting terms in a canonical ``circuit-like'' form. From this form, a circuit representation can be trivially extracted. The circuit representation we compile to here will consist of a fixed number of qubits (which can be referred to by index), a Hadamard gate, and arbitrary multi-controlled phase gates (both zero- and one- controlled with arbitrary angle). The problem of further compiling multi-qubit gates to a finite gate set is well studied~\cite{comp-control-gates,Khattar2025riseofconditionally} and outside the scope of this work.

We begin by defining this canonical form.
\begin{definition}[Normal term]
  Let a \emph{simple} pattern be one in the form:
  \[q_1 \otimes \cdots \otimes q_{n}\]
  where each \(q_i\) is either \(\id{1}\), \(\ket{0}\), \(\ket{1}\), \(\ket{+}\), or \(\ket{-}\). Define a \emph{normal clause} to be a term of the form:
  \[\ifthen q {\Phase{\theta} \otimes \id{k}}\]
  where \(q\) is a simple pattern, and let a \emph{normal term} be a sequential composition of such clauses.
\end{definition}

Normal terms therefore take the following form:
\[
  \ifthen {q_{11} \otimes \cdots \otimes q_{1n}} {\Phase{\theta_1} \otimes \id{j_1}} ; \cdots ; \ifthen {q_{k1} \otimes \cdots \otimes q_{kn}} {\Phase{\theta_k} \otimes{\id{j_k}}}
\]
where each \(q_{ij}\) is either \(\id{1}\), \(\ket{0}\), \(\ket{1}\), \(\ket{+}\), or \(\ket{-}\). A circuit can then be directly extracted from such a form, with each ``if let'' statement being replaced by a multi-controlled phase, conjugated by Hadamard gates on qubits that are plus-controlled or minus-controlled.

To massage terms into such a form, the following cases must be tackled:
\begin{itemize}
\item Tensor products of gates must be reduced to sequential compositions of ``whiskered'' gates, where whiskering refers to taking a tensor product with the identity gate. As an example \(\Gate{Z} \otimes \Gate{Z}\) could be reduced to \(\Gate{Z} \otimes \id{1} ; \id{1} \otimes \Gate{Z}\). The choice to avoid tensor products in our final representation is motivated by tensor products not being stable under control---the control of \(\Gate{Z} \otimes \Gate{Z}\) is not the tensor product of two \(\Gate{CZ}\) gates. In contrast, whiskering and sequential composition \emph{are} stable under control.
\item ``If let'' statements of sequential compositions should be reduced to compositions of ``if let'' statements.
\item Patterns which are unitaries should be reduced to conjugation, for example the swap gate
  \[ \ifthen {(\ifthen {\ket{-1}} {\Phase{\pi}})} {\ifthen {\ket{1-}} {\Phase{\pi}}}\]
  can be reduced to its more usual representation:
  \[ \ifthen {\ket{-1}} {\Phase{\pi}} ; \ifthen {\ket{1-}} {\Phase{\pi}} ; \ifthen {\ket{-1}} {\Phase{\pi}}\]
\item Nested ``if let'' statements can be combined into a single ``if let'' statement with a composed pattern. For example, below the left term can be reduced to the right:
  \[ \ifthen {p} {\ifthen {q} s} \rightsquigarrow \ifthen {p \cdot q} {s} \]
\end{itemize}\todo[kim]{Kim doesn't like the typesetting}

Patterns are not necessarily in the form of a single unitary or a simple pattern, but the evaluation reduces any arbitrary pattern \(p\) to the form \(s \cdot q\), where \(s\) is a unitary term and \(q\) is a simple pattern. This allows normalisation to proceed by conjugating with the unitary \(s\).

We evaluate terms in an \emph{evaluation context} \((q, l, r) : k \to n\), where \(\vdash q : \pattern{m}{n}\) is a simple pattern, and \(l, r \in \mathbb{N}\) such that \(l + k + r = m\).
We motivate our evaluation context as follows: the simple pattern \(q\) allows us to track that the term we are currently evaluating should be applied to a certain subspace, and allows us to evaluate under an ``if let'' statement. The numbers \(l\) and \(r\) record how many identities the term being evaluated has been whiskered with on either side.
\begin{itemize}
\item For a context \((q, l, r) : k \to n\) and unitary term \(s : \unitary{k}\), its evaluation \(\evalu {q,l,r} s\) is a list \([c_1, \dots, c_N]\) where each \(c_i : \unitary{n}\) is a normal clause. For such a list, its \emph{inverse} \([c_1, \dots, c_N]^\dagger = [c_N^\dagger, \dots, c_1^\dagger]\), where \(c_i^\dagger\) is the result of the inversion meta operation defined in \cref{s:syntax}. We write \(c_1 \doubleplus c_2\) for the concatenation of lists \(c_1\) and \(c_2\).
\item For a context \((q,l,r) : k \to n\) and a pattern \(p : \pattern{j}{k}\), its evaluation \(\evalp {q,l,r} p\) is a tuple \(([c_1, \dots, c_N], q')\) where \(q' : l + j + r \to n\) is a simple pattern, and each \(c_i : \unitary{n}\) is a normal clause.
\end{itemize}
For intuition, if \(\evalu {q,l,r} s = [c_1,\dots, c_N]\), then \(c_1 ; \dots; c_N\) should be equivalent to
\[\ifthen {q} {\id{l} \otimes s \otimes \id{r}}\]
  and if \(\evalp {q,l,r} p = ([c_1, \dots, c_N], q')\) then \((c_1 ; \dots; c_N) \cdot q'\) should be equivalent to
  \[q \cdot (\id{l} \otimes p \otimes \id{r})\]
  We make this intuition precise in \cref{s:semantics}.
\begin{definition}[Substitution]
  Given a simple pattern \(q : \pattern{m}{n}\) let \(q[\ket x / i]\) be the result of substituting the \(i\)\textsuperscript{th} \(\id{}\) in \(q\) with \(\ket x\). For example, if \(q = \ket 0 \otimes \id{} \otimes \id{}\) then \(q[\ket - / 0] = \ket 0 \otimes \ket - \otimes \id{}\).
\end{definition}

The evaluation functions are now defined by induction using the rules in \cref{fig:eval}. The normal term of \(s\) can then be extracted by composing the final list of clauses (taking the normal term to be \(\id{n}\) in the empty case, where \(n\) is the number of qubits of the input term).

\begin{figure}\centering

  \begin{mathpar}
    \inferrule{ }{\evalu {q,l,r} {\Phase{\theta}} = [\ifthen q {\Phase{\theta} \otimes \id{l+r}}]}\and
      \inferrule{\evalu {q,l,r} s = c \and \evalu {q,l,r} t = c'}{\evalu {q,l,r} {s ; t} = c \doubleplus c'}\and
      \inferrule{\vdash s : \unitary{k_1} \and \vdash t : \unitary{k_2} \and \evalu {q,l,r+k_2} s = c \and \evalu {q,l+k_1,r} t = c'}{\evalu {q,l,r} {s \otimes t} = c \doubleplus c'}\and
      \inferrule{ }{\evalu {q,l,r} {\id{k}} = []}\and
      \inferrule{\evalp {q,l,r} p = (c, q') \and \evalu {q', l, r} s = c'}{\evalu {q,l,r} {\ifthen p s} = c^\dagger \doubleplus c' \doubleplus c}\and
        \inferrule{x \in \{0, 1, +, -\}}{\evalp {q,l,r} {\ket{x}} = ([], q[\ket{x}/l])}\and
        \inferrule{\evalu {q, l, r} s = c} {\evalp {q,l, r} s = (c, q)}\and
        \inferrule{\evalp {q,l,r} {p_1} = (c, q') \and \evalp {q', l, r} {p_2} = (c', q'')}{\evalp {q,l,r} {p_1 \cdot p_2} = (c' \doubleplus c, q'')}\and
        \inferrule{\vdash p_1 : \pattern {j_1} {k_1} \and \vdash p_2 : \pattern {j_2} {k_2} \and \evalp {q,l,r + k_2} {p_1} = (c, q') \and \evalp {q', l + j_1,r} {p_2} = (c', q'')}{\evalp {q, l, r} {p_1 \otimes p_2} = (c' \doubleplus c, q'')}
      \end{mathpar}
      \caption{Rules for defining the evaluation functions. In each case we assume that \((q, l, r) : k \to n\).}
      \label{fig:eval}
    \end{figure}

We end the section by proving type soundness, which claims that our evaluation function preserves typing judgements.
\begin{theorem}[Type soundness]\label{thm:type-soundness}
  Let \((q, l, r) : k \to n\) be an evaluation context. The following rules are derivable:
  \begin{mathpar}
    \inferrule{\vdash s : \unitary k \and \evalu {q,l, r} s = [c_1, \dots, c_N]} {\vdash c_1; \dots; c_N : \unitary n}\and
    \inferrule{\vdash p : \pattern {m} k \and \evalp {q,l,r} p = ([c_1, \dots, c_N], q')} {\vdash c_1; \dots; c_N : \unitary n \and (q', l, r) : m \to n}
  \end{mathpar}
\end{theorem}
\begin{proof}
  The proof proceeds by simultaneously proving both rules, mutually inducting on terms and patterns. %We defer the details to Appendix~\ref{app:typing}.
\end{proof}

\section{Categorical Semantics}\label{s:semantics}

We now formalise the meaning of our combinator language by equipping it with categorical semantics. Using this we will prove that the evaluation algorithm in \cref{s:compiler} is valid in all models. %Lemmas and theorems throughout this section rely on proofs presented in Appendix~\ref{app:semantics}.

\subsection{Dagger Categories}\label{s:sem-dg}

Except for measurement, quantum programming is intrinsically reversible,
following the rules of quantum mechanics. We capture this reversibility
mathematically by providing, for each morphism $f \colon X \to Y$, an
associated morphism $f^\dagger \colon Y \to X$ in the converse direction.
Formally, a dagger category is a category equipped with an identity-on-objects
involutive contravariant functor $(-)^\dagger \colon \cat C^{\rm op} \to \cat
C$, meaning that $X\dg = X$ and $(fg)\dg = g\dg f\dg$.

\begin{example}
	Here are some examples of relevant dagger categories for reversible
	programming.

	The category with sets as objects and partial injective functions as
	morphisms is a dagger category. We write $\cat{PInj}$ for this category.
	The dagger of a function $f$ is its partial inverse, as in the following
	example morphism $f \colon \{0,1\} \to \{0,1\}$.
	\[
		f(x) = \left\{
			\begin{array}{ll}
				1 				 & \text{if } x=0 \\
				\text{undefined} & \text{otherwise}
			\end{array}
		\right.
		\qquad\qquad
		f^{\dagger} (x) = \left\{
			\begin{array}{ll}
				0 				 & \text{if } x=1 \\
				\text{undefined} & \text{otherwise}
			\end{array}
		\right.
	\]
	This category represents what we refer to as \emph{classical}
	reversibility.

	The category $\cat{Con}$ of Hilbert spaces and contractive linear
	maps---\textit{i.e.}~maps with norm at most $1$---is a dagger category where the dagger is the adjoint of linear
	maps. This category embodies \emph{pure} quantum operations: it contains
	both unitaries and isometries, and so can model pure quantum states as well
	as programs.
\end{example}

Note that the dagger is not necessarily an inverse, but should rather be seen
as a \emph{partial} inverse. This gives us the mathematical leverage to provide
a semantics to ``if let'' (see next subsection). In a dagger category, a
morphism $f$ is a dagger monomorphism (resp. epimorphism) if $f^\dagger f =
\iid$ (resp. if $ff^\dagger = \iid$), and it is a dagger isomorphism if it is
both dagger monic and epic~\cite{heunenvicary:cqm}.

\begin{example}
	In \cat{PInj}, the dagger monomorphisms are the totally defined injective
	functions, and the dagger isomorphisms are the bijections. In $\cat{Con}$,
	the dagger monomorphisms are the isometries and the dagger isomorphisms are
	the unitaries.
\end{example}

Later in this section, we interpret the programs in our languages as dagger
isomorphisms (\emph{e.g.}~unitaries in a quantum setting), and patterns as
dagger monomorphisms (\emph{e.g.}~isometries).

\begin{definition}[Zero morphism]
	A category has \emph{zero morphisms} if for all $X, Y$, there exists
	$0_{X,Y} \colon X \to Y$ such that $0_{X,Y}\dg = 0_{Y,X}$, and $0_{X',Y}
	\circ f = 0_{X,Y} = g \circ 0_{X,Y'}$ for $f \colon X \to X'$, $g \colon Y'
	\to Y$.
\end{definition}

\begin{example}
	In $\cat{PInj}$ the zero morphism $0_{X,Y} \colon X \to Y$ is the function
	defined on no input. In $\cat{Con}$, it is the constantly zero linear map.
\end{example}

A \emph{dagger preserving} functor $F \colon \cat C \to \cat D$ between two
dagger categories is such that $F(f\dg) = F(f)\dg$ for any morphism $f$ in
$\cat C$. If $(C, \otimes, I)$ is symmetric monoidal and $\otimes$ is a dagger
functor, we say that that $(C, \otimes, I)$ is dagger symmetric monoidal.

A \emph{dagger rig category} is a dagger category equipped with symmetric
monoidal structures $(\otimes, I)$ and $(\oplus, O)$, such that $\otimes$ and
$\oplus$ are dagger functors, with their coherence isomorphisms, and with
additional natural dagger isomorphisms (satisfying coherence
conditions~\cite{laplaza:coherence}):
\[
    \begin{array}{c}
        (X \oplus Y) \otimes Z \stackrel{\thicksim}{\longrightarrow} (X \otimes Z) \oplus (Y \otimes Z),
        \quad
        O \otimes X \stackrel{\thicksim}{\longrightarrow} O,
        \\[1.5ex]
        Z \otimes (X \oplus Y) \stackrel{\thicksim}{\longrightarrow} (Z \otimes X) \oplus (Z \otimes Y),
        \quad
        X \otimes O \stackrel{\thicksim}{\longrightarrow} O.
    \end{array}
\]

\begin{example}
	The categories $\cat{PInj}$ and $\cat{Con}$ are dagger rig categories. The
	monoidal structures are the usual ones for $\cat{PInj}$: the product is the
	product of sets, its unit is a singleton, the sum is the disjoint union of
	sets, and its unit is the empty set. In $\cat{Con}$, the product is the
	tensor product, and its unit is a one-dimensional Hilbert space. Its sum is
	the direct sum of spaces, and the unit is the zero-dimensional Hilbert
	space $\{ 0 \}$.
\end{example}

The rig structure plays a central part in our semantics. We later
make sense of the subspaces pointed out by patterns in the language with the
$\oplus$ tensor, through the notion of \emph{independent coproducts} (see
\cref{def:ind-coproduct}). In this view, for example, $X_1$ and $X_2$ are some
subspaces of $X_1 \oplus X_2$. If we interpret qubits as $I \oplus I$, this
means that the space generated by $\ket 0$ (resp. $\ket 1$) are subspaces---and
they are not the only ones.

Now consider programs that act on the space $(X_1 \oplus X_2) \otimes Y$. The
rig structure allows us to identify $(X_1 \oplus X_2) \otimes Y$ with $(X_1
\otimes Y) \oplus (X_2 \otimes Y)$, which means that $X_1 \otimes Y$ and $X_2
\otimes Y$ are some subspaces of $(X_1 \oplus X_2) \otimes Y$. This identifies subspaces
with patterns of the form $p \otimes \id n$.

We later prove that we can derive a dagger rig structure with independent
coproducts (see \cref{def:ind-coproduct}) that are preserved by the tensor
product (see \cref{def:preserved-ic}).

\subsection{Independent Coproducts}\label{s:sem-is}

Our syntax in Section~\ref{s:syntax} relies on the ``if let'' conditional,
which is an instance of case splitting. Case splitting in programming languages
is usually interpreted with a \emph{disjointness} structure~\cite{coproducts}.
For example, if the case splitting happens on a Boolean, the set of potential
results is $X_1 \sqcup X_2$, where $X_1$ happens when the Boolean is true, and
$X_2$ when it is
false.
In category theory, we capture this behaviour with coproducts, which are cospans:
\(
	\begin{tikzcd}[cramped, ampersand replacement=\&, column sep=1cm]
		X_1 \& X_1 + X_2 \& X_2
		\arrow["i_1", from=1-1, to=1-2]
		\arrow["i_2"', from=1-3, to=1-2]
	\end{tikzcd}
\)
such that for all $f_1 \colon X_1 \to Y$ and $f_2 \colon X_2 \to Y$, there is a
unique mediating morphism $m \colon X_1 + X_2 \to Y$ such that the
diagram
\begin{equation}\label{eq:coprod}
	\begin{tikzcd}[ampersand replacement=\&, column sep=2cm]
		X_1 \& X_1 + X_2 \& X_2 \\
		\& Y \&
		\arrow["i_1", from=1-1, to=1-2]
		\arrow["i_2"', from=1-3, to=1-2]
		\arrow["f_1"', from=1-1, to=2-2]
		\arrow["f_2", from=1-3, to=2-2]
		\arrow["m", dashed, from=1-2, to=2-2]
	\end{tikzcd}
\end{equation}
commutes, without any special conditions on $f_1$ and $f_2$. In our setting, we also need to ensure our program remains reversible and has a reversible
semantics. In the particular case of case splitting, this reversibility
condition requires that we decide deterministically whether the result of type $Y$
through $m$ came from the map $f_1$ or the map $f_2$.

To do so, we introduce a notion of independence of morphisms, characterised by
the condition $f_1^\dagger f_2 = 0$ (see Definition~\ref{def:ind-cospan}),
where $f_1^\dagger$ is thought of as the (partial) inverse of $f_1$. This fits
the notion of compatibility of morphisms in classical reversible
settings~\cite[Definition 8]{kaarsgaardetal:join} and the one for linear maps
between Hilbert spaces~\cite[Section 2.3]{daveetal:control}.

\begin{definition}[Independent cospan]\label{def:ind-cospan}
	In a dagger category with zero morphisms, an \emph{independent cospan} is a pair
	of morphisms $(f_1 \colon X_1 \to Y, f_2 \colon X_2 \to Y)$ such that
	$f_1^\dagger f_2 = 0$.
\end{definition}

Additionally, we introduce a universal equivalent to independent cospans. In
the same vein as coproducts, an \emph{independent coproduct} is an independent
cospan for which there exists a mediating morphism with any other independent
cospan. An independent coproduct is also jointly epic, ensuring uniqueness of
not only mediating morphisms, but of all morphisms that make the diagrams such
as (\ref{eq:coprod}) commute. This uniqueness is essential to prove that our
programs are indeed interpreted as dagger isomorphisms.

\begin{definition}[Independent coproduct]\label{def:ind-coproduct}
	An \emph{independent coproduct} in a dagger category with zero morphisms is
	a jointly epic independent cospan $(\iota_1 \colon X_1 \to X, \iota_2
	\colon X_2 \to X)$ such that $\iota_1$ and $\iota_2$ are dagger
	monomorphisms, and for all independent cospans $(f_1 \colon X_1 \to Y, f_2
	\colon X_2 \to Y)$, there is a unique morphism $u \colon X \to Y$, making
	the following diagrams commute.
	\[
		\begin{tikzcd}[ampersand replacement=\&, row sep=7mm, column sep=2cm]
			X_1 \& X \& X_2  \& X_1 \\
			\& Y \& X \& X_2 \& Y
			\arrow["\iota_1", "\shortmid"{marking}, tail, from=1-1, to=1-2]
			% \arrow["0"', bend left=30, from=1-1, to=1-3]
			\arrow["f_1"', from=1-1, to=2-2]
			\arrow["{u}", dashed, from=1-2, to=2-2]
			\arrow["\iota_2"', "\shortmid"{marking}, tail, from=1-3, to=1-2]
			\arrow["f_2", from=1-3, to=2-2]
      \arrow["\iota_1"', from=1-4, to=2-3]
      \arrow["0", from=1-4, to=2-4]
      \arrow["f_1", from=1-4, to=2-5]
      \arrow["\iota_2^\dag"', from=2-3, to=2-4]
      \arrow["f_2^\dag", from=2-5, to=2-4]
		\end{tikzcd}
	\]
\end{definition}

\begin{lemma}\label{lem:ic-choice}
	Let $p$ be a morphism. If both $(p,c)$ and $(p,d)$ are independent
	coproducts, then $cc^\dagger = dd^\dagger$.

	It also implies that that $d\dg c$ is a dagger isomorphism. Therefore
	$(p,c)$ and $(p,d)$ are the same independent coproduct up to dagger isomorphism.
\end{lemma}

\begin{definition}[Independent coproducts]
	A category \emph{has independent coproducts} if for all pairs of
	objects $X_1, X_2$, there exists an object $X_1 \oplus X_2$ and chosen
	morphisms $\iota_i \colon X_i \to X_1 \oplus X_2$ such that $(\iota_1,
	\iota_2)$ is an independent coproduct. Given an independent cospan $(f_1,
	f_2)$, its mediating morphism with the chosen independent coproduct is
	written $[f_1, f_2]$.
\end{definition}

\begin{theorem}\label{th:ic-smc}
	If \(\cat C\) has independent coproducts and a zero object $O$, it is
	symmetric monoidal.
\end{theorem}

\begin{example}
	Both $\cat{PInj}$ and $\cat{Con}$ have independent coproducts, given by their usual notion of direct sum.
\end{example}

In our language, we also need to interpret the tensor
product, which manifests as a monoidal product. If this monoidal product is
compatible with independent coproducts (as described below), we obtain a dagger
rig category.

\begin{definition}[Preservation of independent coproducts]\label{def:preserved-ic}
	We say that $(\cat C, \otimes, I)$ \emph{preserves independent coproducts} if for
	all independent coproducts $(\iota_1, \iota_2)$ and objects $Y$, then
	$(\iota_1 \otimes \iid_Y, \iota_2 \otimes \iid_Y)$ is also an independent
	coproduct.
\end{definition}

\begin{theorem}\label{th:rig}
	If $(\cat C, \otimes, I)$ is a dagger symmetric monoidal category with a
	monoidal zero object $O$ (namely, equipped with dagger isomorphisms $O
	\otimes X \cong O$), and independent coproducts preserved by the monoidal
	structure, then $\cat C$ is a dagger rig category.
\end{theorem}

\begin{example}
	Our two examples $\cat{PInj}$ and $\cat{Con}$ have such a monoidal
	structure, and with independent coproducts, are rig categories under the
	conditions of Theorem~\ref{th:rig}. In fact, the category $\cat{PInj}$ is
	the standard one to model classical reversible
	programming~\cite{kaarsgaard2021join, kostia2024invrec}. It has many
	characteristics in common with $\cat{Con}$, as highlighted by their
	axioms~\cite{heunen2024axioms}. However, Axiom~(5) for $\cat{Con}$ shows
	that `mixture occurs': there is a map $I \to I \oplus I$ that is orthogonal
	to neither injection; physicists call this `superposition'. In this, we
	present $\cat{Con}$ as suitable for quantum computing, as opposed to
	$\cat{PInj}$.
\end{example}

\subsection{Semantics of ``if let''}\label{s:sem-if}

Let $(\cat{C}, \otimes, I)$ be a dagger symmetric monoidal category with a
monoidal zero object and independent coproducts, equipped with a chosen,
\emph{distinguished} independent coproduct
\(
	\begin{tikzcd}[ampersand replacement=\&, cramped]
		I \& I \oplus I \& I
		\arrow["r_1", "\shortmid"{marking}, tail, from=1-1, to=1-2]
		\arrow["r_2"', "\shortmid"{marking}, tail, from=1-3, to=1-2]
	\end{tikzcd}
\)
and a family of scalars $\phi_\theta \colon I \to I$ for $\theta \in \mathbb R$
satisfying $\phi_\theta \circ \phi_{\theta'} = \phi_{\theta + \theta'}$,
$\phi_\theta\dg = \phi_{-\theta}$ and $\phi_0 = \iid_I$.

\begin{remark}
	We choose here to have a phase group parameterised by real numbers, helping
	for a simpler presentation of both the syntax and its semantics. Note that
	the approach to phase can be more refined, and we could parameterise with
	any group.
\end{remark}

We provide a semantics of our language in the category $\cat C$. We
first fix the semantics for qubit types as: $\sem{\qreg 0} = I,$ and
$\qsem{(n{+}1)} = \qsem n \otimes (I \oplus I)$. Note that, due to the rig
structure (see \cref{th:rig}), there is an isomorphism $\qsem{(n{+}m)}
\cong \qsem n \otimes \qsem m$ for all $n,m$.  In the rest of the section, we
allow ourselves to write ``$\cong$'' for coherence morphisms, to keep apparent
only the key points of the semantics.  Both well-formed terms $\unitary n$
and patterns $\pattern n m$ are interpreted as morphisms $\qsem n \to
\qsem n$ and $\qsem n \to \qsem m$, and we later prove that the interpretation
of a term is a dagger isomorphism, and the one of a pattern is dagger monic
(see Theorem~\ref{th:well-defined-sem}).

\begin{remark}
	The data types used in the syntax are only qubits, and the semantics of a
	qubit data is $I \oplus I$. While we keep this restriction, the content of
	this section is generalisable to any size of data (for example qutrits,
	modelled as $I \oplus I \oplus I$). This would, however, require a
	different syntax.
\end{remark}

The full semantics of the language is detailed in Figure~\ref{fig:sem}. The
cornerstone of the language is the ``if let'' expression, and the semantics of
the other terms follow easily. Given a pattern $\vdash p \colon \pattern n m$,
assume that $\sem{\vdash p \colon \pattern n m}$ is part of an independent
coproduct
\(
	\begin{tikzcd}[cramped]
		\qsem n & \qsem m & \bullet\ .
		\arrow["\sem p", "\shortmid"{marking}, tail, from=1-1, to=1-2]
		\arrow["\ \sem p^\perp"', "\shortmid"{marking}, tail, from=1-3, to=1-2]
	\end{tikzcd}
\)
The choice of $\sem p^\perp$ does not matter, up to dagger isomorphism. With a
unitary $\vdash s \colon \unitary n$, we have an independent cospan $(\sem p
\sem s, \sem p^\perp)$ whose mediating morphism with $(\sem p, \sem p^\perp)$
is written \(\left[\sem p \sem s, \sem p^\perp\right]\) as shown below, which we let be the semantics of
$\vdash \ifthen p s \colon \unitary m$.
\begin{equation}\label{eq:iflet-sem}
	\begin{tikzcd}[ampersand replacement=\&, row sep=10mm, column sep=1.5cm]
		\qsem n \& \qsem m \& \ \bullet\  \\
		\& \qsem m
		\arrow["\sem p"', "\shortmid"{marking}, tail, from=1-1, to=1-2]
		% \arrow["0"', bend left=30, from=1-1, to=1-3]
		\arrow["\sem p \sem s"', from=1-1, to=2-2]
		\arrow["{\left[\sem p \sem s, \sem p^\perp\right]}"{description}, dashed, from=1-2, to=2-2]
		\arrow["\ {\sem p^\perp}", "\shortmid"{marking}, tail, from=1-3, to=1-2]
		\arrow["{\sem p^\perp}", from=1-3, to=2-2]
	\end{tikzcd}
	\qquad
	\begin{tikzcd}[ampersand replacement=\&, row sep=10mm, column sep=1.5cm]
		\qsem n \& \qsem m \& \ \bullet\  \\
		\qsem n \& \qsem m \& \ \bullet\
		\arrow["\sem p"', "\shortmid"{marking}, tail, from=1-1, to=1-2]
		% \arrow["0"', bend left=30, from=1-1, to=1-3]
		\arrow["\sem s"', from=1-1, to=2-1]
		\arrow["\sem p"', from=2-1, to=2-2]
		\arrow["{\left[\sem p \sem s, \sem p^\perp\right]}"{description}, dashed, from=1-2, to=2-2]
		\arrow["\ {\sem p^\perp}", "\shortmid"{marking}, tail, from=1-3, to=1-2]
		\arrow["\iid", from=1-3, to=2-3]
		\arrow["{\sem p^\perp}", from=2-3, to=2-2]
	\end{tikzcd}
\end{equation}
In other words, the dagger monomorphism $\sem p$ works as a subobject, on which
$\sem s$ is applied. The resulting $\left[\sem p \sem s, \sem p^\perp\right]$ is the morphism that acts as $\sem s$ on
the subobject identified as $\sem p$, and that acts as the identity on its
orthogonal subobject $\spp$.

\begin{figure}[t]
	\begin{align*}
		\sem{\vdash \Phase\theta \colon \unitary 0} &= \phi_\theta \\
		\sem{\vdash \id n \colon \unitary n} &= \iid_{\qsem n} \\
		\sem{\vdash s;t \colon \unitary n} &= \sem{\vdash t \colon \unitary n}
		\circ \sem{\vdash s \colon \unitary n} \\
		\sem{\vdash s \otimes t \colon \unitary{(n{+}m)}} &= ~\cong \circ \left(
		\sem{\vdash s \colon \unitary n} \otimes \sem{\vdash t \colon \unitary n}
		\right) \circ \cong \\
		\sem{\vdash \ifthen p s \colon \unitary m} &= \left[\sem p \sem s, \sem p^\perp\right] \\[1.6ex]
		%  	\left[ \sem{\vdash p \colon \pattern n m} \sem{\vdash s \colon \unitary n},
		%  	\sem{\vdash p \colon \pattern n m}^\perp \right] \circ \left[ \sem{\vdash p
		%  	\colon \pattern n m}, \sem{\vdash p \colon \pattern n m}^\perp \right]
		%  	^\dagger \\
		\sem{\vdash \ket 0 \colon \pattern 0 1} &= \iota_1
		\qquad\qquad\
		\sem{\vdash \ket 1 \colon \pattern 0 1} = \iota_2 \\
		\sem{\vdash \ket + \colon \pattern 0 1} &= r_1
		\qquad\qquad
		\sem{\vdash \ket - \colon \pattern 0 1} = r_2 \\
		\sem{\vdash s \colon \pattern n n} &= \sem{\vdash s \colon \unitary n} \\
		\sem{\vdash p \cdot q \colon \pattern l m} &= \sem{\vdash p \colon \pattern
		n m} \circ \sem{\vdash q \colon \pattern l n} \\
		\sem{\vdash p \otimes q \colon \pattern{(n{+}n')}{(m{+}m')}} &= ~\cong
		\circ \left( \sem{\vdash p \colon \pattern n m} \otimes \sem{\vdash q
		\colon \pattern{n'}{m'}} \right) \circ \cong
	\end{align*}
	\caption{Categorical semantics, defined by induction of the type derivation.}
	\label{fig:sem}
\end{figure}

We prove, by induction on the typing derivations, that our semantics is
well-defined. In particular, we show that programs $s$ are interpreted as
dagger isomorphisms (meaning that $\sem s\dg \sem s = \iid$ and $\sem s\sem
s\dg = \iid$) and that patterns $p$ are interpreted as dagger monomorphisms
(meaning that $\sem p\dg \sem p = \iid$).  The well-definedness also heavily
relies on the existence of an \emph{orthogonal} morphism to $\sem p$ for all
$p$, in order to compute the ``if let'' statement as in (\ref{eq:iflet-sem}).

\begin{theorem}\label{th:well-defined-sem}
	For any well-typed pattern $p$ and unitary $s$, we have that:
	\begin{itemize}
		\item there exists a morphism $\sem p ^\perp$ such that $\left( \sem p,
			\sem p ^\perp \right)$ is an independent coproduct;
		\item the morphism $\spp$ is unique up to dagger isomorphism;
		\item the morphism $\sem p$ is a dagger monomorphism;
		\item the morphism $\sem s$ is a dagger isomorphism.
	\end{itemize}
\end{theorem}

Note that given a well-typed pattern $p$, we can give a formula for $\spp$, as follows:
\[
	\begin{array}{c}
		\sem{\ket 0}^\perp = \sem{\ket 1}
		\quad
		\sem{\ket +}^\perp = \sem{\ket -}
		\quad
		\sem s ^\perp = 0
		\\[1.3ex]
		\sem{p \cdot q}^\perp = \left[\sem p \sem q^\perp, \sem p^\perp \right]
		\circ \cong
		\quad
		\sem{p \otimes q}^\perp = ~\cong\circ \left[\sem p^\perp \otimes \sem
		q, \left[ \sem p \otimes \sem q^\perp, \sem p^\perp \otimes \sem
		q^\perp \right] \right] \circ\cong
	\end{array}
\]
in which our use of $\cong$ contains only coherence
morphisms~\cite{laplaza:coherence} that are identities in after the
semi-strictification of rig categories (namely, all coherence morphisms except
the multiplicative symmetry and the right distributor).

\begin{remark}
	The category $\cat{PInj}$ is a (degenerate) model of the language. It has
	all the structure required, with $\phi_\theta = \iid_I$ for all $\theta$,
	and $r_i = \iota_i$ for $i \in \{1,2\}$.
	Naturally, the best model for our language is $\cat{Con}$, where patterns are
	interpreted as isometries, and programs as unitaries.
\end{remark}

%\subsection{Properties}\label{sub:sem-prop}

We showcase some equations that hold in any model of the language.

\begin{proposition}\label{prop:equations}
	If $p$ and $q$ are well-typed patterns, and $s$ and $t$ are well-typed terms,
	then we have:
	\begin{itemize}
		\item
			\(
				\sem{\ifthen p {\ifthen q s}} = \sem{\ifthen {p \cdot q} s} \text;
			\)
		\item
			\(
				\sem{\ifthen p {(s;t)}} = \sem{\ifthen p s ; \ifthen p t} \text;
			\)
		\item
			\(
				\sem{\ifthen t s} = \sem t \sem s \sem t\dg \text.
			\)
	\end{itemize}
\end{proposition}

The categorical dagger agrees with the syntactic inverse (see
Definition~\ref{def:syn-inverse}).

\begin{proposition}
	Given a well-typed term $s$, we have that $\sem s \dg = \sem{s\dg}$.
\end{proposition}

\subsection{Soundness}\label{sub:soundness}

We now show that any categorical model $\cat C$ is sound with respect to the
evaluation function for the compilation. We prove that the compilation
procedure does not alter the semantics of the program.

Each normal clause $c_i$ already has a well-defined semantics (see
Theorem~\ref{th:well-defined-sem}), and we define the semantics of a list of
clauses $[c_1, \dots, c_l]$ as the composition $\sem{c_l} \circ \dots \circ
\sem{c_1}$. In the following, we allow ourselves to loosely write $c$ for a
list of clauses, and $\sem c$ for its semantics in a categorical model $\cat
C$.
Since $q$ designates the subspace on which the program is applied, we should
have that if $\evalu {q,l,r} s = c$, then $\sem c$ is the mediating morphism below:
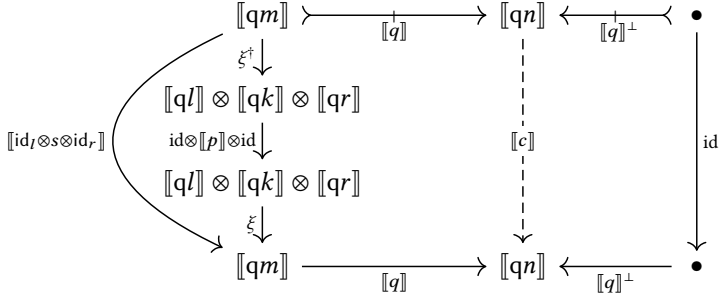
\begin{equation*}
	\begin{tikzcd}[ampersand replacement=\&, column sep=1.5cm, row sep=5mm]
		\qsem m \& \qsem n \& \ \bullet\  \\
		\qsem l \otimes \qsem k \otimes \qsem r  \&\& \\
		\qsem l \otimes \qsem k \otimes \qsem r  \&\& \\
		\qsem m \& \qsem n \& \ \bullet\
		\arrow["\sem q"', "\shortmid"{marking}, tail, from=1-1, to=1-2]
		% \arrow["0"', bend left=30, from=1-1, to=1-3]
		\arrow["\sem c"{description}, dashed, from=1-2, to=4-2]
		\arrow["\ {\sem q^\perp}", "\shortmid"{marking}, tail, from=1-3, to=1-2]
		\arrow["\sem q"', from=4-1, to=4-2]
		\arrow["{\sem q^\perp}", from=4-3, to=4-2]
		\arrow["\sem{\id l \otimes s \otimes \id r}"', curve={height=70pt}, from=1-1, to=4-1]
		\arrow["\xi\dg"', from=1-1, to=2-1]
		\arrow["\iid \otimes \sem p \otimes \iid"', from=2-1, to=3-1]
		\arrow["\xi"', from=3-1, to=4-1]
		\arrow["\iid", from=1-3, to=4-3]
	\end{tikzcd}
\end{equation*}
where $\xi \colon \qsem l \otimes \qsem k \otimes \qsem r \to \qsem m$ is the
coherence isomorphism that rewrites one object into the other, since we know
that $l + k + r = m$, but $\xi$ basically acts as the identity. Intuitively, we
get that the semantics of $\evalu {q,l,r} s$ is equal to the one of the term:
$\ifthen q {(\id l \otimes s \otimes \id r)}$.

In the case of $\evalp {q,l,r} p = (c,q')$, the intuition is that the semantics
of the pattern $c \cdot q'$ is equal to the one of $q \cdot (\id l \otimes p
\otimes \id r)$; and their orthogonal morphisms are equal as well.

\begin{lemma}[Induction hypothesis]
	We have the following:
	\begin{itemize}
		\item If $\evalu {q,l,r} s = c$, then $\sem c$ is the mediating morphism in the diagram above; in other words, the unique morphism such that:
			\[
				\left\{ \begin{array}{l}
					\sem c \sem{q} = \sem q \xi (\iid \otimes \sem s \otimes \iid) {\xi}\dg \\
					\sem c \sem{q}^\perp = \sem{q}^\perp \text;
				\end{array}\right.
			\]
		\item if $\evalp {q,l,r} p = (c,q')$, then we have that:
			\(
				\sem c \sem{q'} = \sem q \xi (\iid \otimes \sem p \otimes \iid) {\xi'}\dg \text.
			\)
	\end{itemize}
\end{lemma}

Soundness is a direct corollary: we can conclude that the compilation
scheme described by the evaluation function does not alter the program.

\begin{theorem}[Soundness]\label{th:soundness}
	Given a well-formed unitary term $s$, we have $\sem{s} = \sem{\evalu
	{\id{}^{\otimes n}, 0, 0} s}$.
\end{theorem}

% \todo[louis,inline]{Soundness is a bit weak in my opinion, because the terminal
% category is a model and is sound, for example. A stronger statement is some
% sort of adequacy, meaning:
% \begin{center}
% 	if $\sem s = \sem t$, then $\mathsf{eval}~s \equiv \mathsf{eval}~t$
% \end{center}
% with the right definition for $\equiv$. This would, of course, only be true for
% $\sem -$ being the model of unitaries, which is fine (and meaningful!).
%
% I don't think it is worth spending a crazy amount of time on this. My current
% idea is to point out that it is possible to obtain a complete equational theory
% for $\equiv$ \cite{clementetal:equationalcircuits}, and adequacy would follow
% directly.
% }

\section{Beyond Combinators}
\label{sec:beyond-combinators}

In this final section, we give two nominal variations of the language, exploring how the ``if let'' construction could manifest in different contexts.
\begin{itemize}
\item A functional language:
  here we allow qubits to be bound to linear variables, and unitaries are treated as functions which consume input qubits, producing new output qubits.
  This allows a more concise syntax where a unitary can be applied to a subset of the available qubits without inserting explicit identity unitaries.
  A consequence of this format is that there arises a native implementation of the ``swap'' gate on two qubits, which in turn allows swaps to appear within ``if let'' blocks.
\item An imperative language:
  in the imperative model, variables are viewed as representing the \emph{location} of a qubit (i.e. their denotation is given by an inclusion into some ambient space) instead of viewing variables as storing the state of a qubit at some fixed point in time.
  In this model, variables need not be treated linearly, and unitaries are given by expressions rather than functions, and are viewed as mutating the variables they act on, rather than consuming them.
\end{itemize}

The following subsections briefly sketch these languages, exploring their differences to the core combinator language.
\subsection{Functional}\label{s:functional}

In the functional variant, we introduce variables which intuitively store the quantum state at intermediate stages of the program.
Gates can then be explicitly applied to any specific qubits available, rather than requiring them to be tensored with appropriate identity operations to apply them to specific qubits.

With access to variables, the typing judgements must now be parameterised by a context. As unitaries are no longer the primitive constructions, we remove the unitary type \(\unitary{n}\), and instead have the following typing judgement for terms:
\[ \Gamma \vdash s : \qreg{n} \]
The context \(\Gamma\) is a list \(x_1 : \qreg{n_1}, x_2 : \qreg{n_2}, \dots\) where the \(x_i\) are variable names. We write \(\cdot\) for the empty context and \(\Gamma, \Delta\) for the concatenation of contexts \(\Gamma\) and \(\Delta\). As is usual in the literature, we treat terms as equal up to \(\alpha\)-equivalence.

Instead of the composition operator present in the combinator language, we now have a ``let'' expression, allowing us to bind the (possibly multiple) outputs of a term. In general this will have the form:
\[\letin{x_1 \otimes \cdots \otimes x_n = s}t\]
We refer to expressions of the form \(x_1 \otimes \cdots \otimes x_n\) as \emph{copatterns}, as they are also exactly the terms that can be matched to a pattern in an ``if let'' statement. To make this precise, we specify three classes of syntax: \(\mathsf{Copattern} \subset \mathsf{Unitary} \subset \mathsf{Pattern}\). All three syntax classes are generated inductively, from the following grammars:
\begin{alignat*}{4}
  &\mathsf{Copattern}\text{:}~~~~&\ &c&\ &\Coloneq&\ &x \mid c_1 \otimes c_2\\
  &\mathsf{Unitary}\text{:}&&s, t&&\Coloneq&&x \mid s \otimes t \mid \Phase{\theta} \mid \letin{c = s}{t} \mid \ifthen{p = c}{s}\\
  &\mathsf{Pattern}\text{:}&&p, q&&\Coloneq&&x \mid p \otimes q \mid \Phase{\theta} \mid \letin{c = p}{q} \mid \ifthen{p = c}{s} \mid \ket{0} \mid \ket{1} \mid \ket{+} \mid \ket{-}
\end{alignat*}
where \(x\) represents a variable.

Within this syntax, we treat variables linearly, which corresponds to the deletion and copying of qubits being forbidden. This is reflected in the typing rules for terms, found below.
\begin{mathpar}
  \inferrule{ }{x : \qreg{n} \vdash x : \qreg{n}}\and
  \inferrule{ }{\cdot \vdash \ket{0} : \qreg{1}}\and
  \inferrule{ }{\cdot \vdash \ket{1} : \qreg{1}}\and
  \inferrule{ }{\cdot \vdash \ket{+} : \qreg{1}}\and
  \inferrule{ }{\cdot \vdash \ket{-} : \qreg{1}}\and
  \inferrule{\Gamma \vdash s : \qreg{n} \and \Delta \vdash t : \qreg{m}}{\Gamma, \Delta \vdash s \otimes t : \qreg{(n + m)}}\and
  \inferrule{ }{\cdot \vdash \Phase{\theta} : \qreg{0}}\and
  \inferrule{\Theta \vdash s : \qreg{m} \and \Delta \vdash c : \qreg{m} \and \Gamma, \Delta \vdash t : \qreg{n}}{\Gamma, \Theta \vdash \letin{c = s}t : \qreg{n}}\and
  \inferrule{\Delta \vdash c : \qreg{m} \and \Theta \vdash p : \qreg{m} \and \Theta \vdash s : \qreg{n}}{\Delta \vdash \ifthen{p = c}{s} : \qreg{m}}
\end{mathpar}
With the above typing rules the term
\[\letin{x' \otimes z' = \left( \ifthen{\ket{1} \otimes \ket{-} = x \otimes z}{\Phase{\pi}} \right)}x' \otimes y \otimes z'\]
in the context \(x : \qreg{1}, y : \qreg{1}, z : \qreg{1}\), which applies a \(\Gate{CX}\) gate to the first and third qubits, is not well-typed, as it uses variables ``in the wrong order''. To remedy this, the following exchange rule can be added.
\begin{mathpar}
  \inferrule{\Gamma, x : \qreg{n}, y: \qreg{m}, \Delta \vdash s : \qreg{l}}{\Gamma, y: \qreg{m}, x : \qreg{n}, \Delta \vdash s : \qreg{l}}
\end{mathpar}

Introducing this exchange rule has an unintended side effect; swaps (and therefore their controlled variants) are now natively representable within the language. For example:
\[x : \qreg{1}, y : \qreg{1}, z: \qreg{1} \vdash \ifthen{\ket{1} \otimes y' \otimes z' = x \otimes y \otimes z}{z' \otimes y'} : \qreg{3}\]
Performs a swap of \(y\) and \(z\), controlled on \(x\).

In the term above, we were required to ``rebind'' \(y\) and \(z\) to \(y'\) and \(z'\), as this functional variant does not allow ``variable capture'': the use of variables in the body of the ``if let'' which were defined outside the pattern. In principle, the typing rules could be modified to allow this, however doing so creates ambiguity in the semantics, as it can be unclear whether swaps occur inside the ``if let'' block (possibly being controlled) as opposed to happening before.

\subsection{Semantics of the Functional Language}\label{s:sem-functional}
Let $(\cat{C}, \otimes, I)$ be a dagger symmetric monoidal category with a zero
object (such that $O \otimes X \cong O$ for all $X$) and independent
coproducts, equipped with a chosen, \emph{distinguished} independent coproduct
\(
	\begin{tikzcd}[ampersand replacement=\&, cramped]
		I \& I \oplus I \& I
		\arrow["r_1", "\shortmid"{marking}, tail, from=1-1, to=1-2]
		\arrow["r_2"', "\shortmid"{marking}, tail, from=1-3, to=1-2]
	\end{tikzcd}
\)
and a family of scalars $\phi_\theta \colon I \to I$ for $\theta \in \mathbb R$
satisfying $\phi_\theta \circ \phi_{\theta'} = \phi_{\theta + \theta'}$,
$\phi_\theta\dg = \phi_{-\theta}$ and $\phi_0 = \iid_I$.

The semantics of a judgement $\Gamma \vdash p \colon \qreg n$ is given as a
morphism
\(
	\sem{\Gamma \vdash p \colon \qreg n} \colon \sem\Gamma \to \sem{\qreg n}
\)
where $\sem\Gamma = \sem{\qreg{n_1}} \otimes \dots \otimes \sem{\qreg{n_k}}$
for $\Gamma = x_1 \colon \qreg{n_1}, \dots, x_k \colon \qreg{n_k}$. Note that
given the latter $\Gamma$, there is a canonical copattern $c_\Gamma = x_1
\otimes \dots \otimes x_k$, such that $\Gamma \vdash c_\Gamma \colon
\qreg{(n_1{+}\dots{+}n_k)}$.

We define the following semantics:
\begin{align*}
	\sem{x \colon \qreg{n} \vdash x \colon \qreg{n}} &= \iid_{\qsem n} \\
 	\sem{\cdot \vdash \ket{0} \colon \qreg{1}} &= \iota_1
	\qquad\qquad\
 	\sem{\cdot \vdash \ket{1} \colon \qreg{1}} = \iota_2 \\
 	\sem{\cdot \vdash \ket{+} \colon \qreg{1}} &= r_1
	\qquad\qquad
 	\sem{\cdot \vdash \ket{-} \colon \qreg{1}} = r_2 \\
  	\sem{\cdot \vdash \Phase{\theta} \colon \qreg{0}} &= \phi_{\theta} \\
 	\sem{\Gamma, \Delta \vdash s \otimes t \colon \qreg{(n + m)}} &= \sem{\Gamma
 	\vdash s \colon \qreg{n}} \otimes \sem{\Delta \vdash t \colon \qreg{m}} \\
  	\sem{\Gamma, \Theta \vdash \letin{c = s}t \colon \qreg{n}} &=
	\sem{\Gamma, \Delta \vdash t \colon \qreg{n}} \circ \left( \iid_{\sem\Gamma}
	\otimes \left( \sem{\Delta \vdash c \colon \qreg{m}}^\dagger \circ \sem{\Theta
	\vdash s \colon \qreg{m}} \right) \right) \\
	\sem{\Delta \vdash \ifthen{p = c}{s} \colon \qreg{m}} &= u \circ
	\sem{\Delta \vdash c \colon \qreg{m}}
\end{align*}
where $u$ is the mediating morphism in:
\begin{equation}\label{eq:fun-iflet-sem}
	\begin{tikzcd}[ampersand replacement=\&, row sep=6mm, column sep=1.5cm]
		\sem \Theta \& \qsem m \& \ \bullet\  \\
		\qsem n \& \& \\
		\sem \Theta \& \qsem m \& \ \bullet \
		\arrow["\sem p"', "\shortmid"{marking}, tail, from=1-1, to=1-2]
		\arrow["\sem p", from=3-1, to=3-2]
		% \arrow["0"', bend left=30, from=1-1, to=1-3]
		\arrow["\sem s"', from=1-1, to=2-1]
		\arrow["\sem{c_\Theta}"', from=2-1, to=3-1]
		\arrow["\iid", from=1-3, to=3-3]
		\arrow["u"{description}, dashed, from=1-2, to=3-2]
		\arrow["\ {\sem p^\perp}", "\shortmid"{marking}, tail, from=1-3, to=1-2]
		\arrow["{\sem p^\perp}"', from=3-3, to=3-2]
	\end{tikzcd}
\end{equation}

We expect to have the following results, with a similar proof as for
Theorem~\ref{th:well-defined-sem}.
\begin{itemize}
	\item If we have $\Gamma \vdash p \colon \qreg n$ with $p$ a pattern,
		then $\sem{\Gamma \vdash p \colon \qreg n}$ is an isometry.
	\item If we have $\Gamma \vdash c \colon \qreg n$ with $c$ a unitary,
		then $\sem{\Gamma \vdash c \colon \qreg n}$ is a unitary.
\end{itemize}

\subsection{Imperative}\label{s:imperative}

The final version of the syntax is an imperative syntax. By exploiting that unitaries always return the same number of qubits as they are input, we can instead view a unitary as a procedure that \emph{mutates} its input qubits, and has no return value.

In this syntax, quantum computation is \emph{effectful}. Under this view, our qubit variables now act as qubit ``locations'' rather than qubit states, and can be freely copied or unused, allowing the linearity condition to be dropped. We note that in this syntax, copying a variable is more akin to creating an alias, rather than duplicating a qubit state (which cannot be done by the no-cloning theorem).

Removing the need for return values comes with various advantages:
\begin{itemize}
\item The overall syntax becomes more concise.
\item It is no longer possible to natively represent swaps as in the functional variant, yet we do not have a rigid qubit order as in the combinator variant.
\item It is much easier to give a typing rule to the ``if let'' clause which allows ``variable capture'', the use of variables defined outside the pattern.
\end{itemize}
To demonstrate the last point, we consider the following program:
\[ x : \qreg{1}, y : \qreg{1} \vdash (\ifthen{\ket{1} = y}{\Gate{X}[x]}); \Gate{H}[x]\]
We are free here to use the variable \(x\) in the body of the ``if let'' clause, despite it not appearing in the pattern \(\ket{1}\). Contrary to the functional variation of the syntax, the variable \(x\) is not treated linearly, and so we are free to continue using \(x\) after the ``if let'' clause. We also note that this program can be viewed as using the variables \(x\) and \(y\) ``in the wrong order''. The translation of this program into the functional syntax would require multiple uses of the exchange typing rule, yet in this syntax the program can be naturally interpreted without any swap gates.

A disadvantage of the imperative approach is its treatment of patterns. Patterns naturally have a more functional treatment, causing the resulting language to contain functional and imperative components. For example, in the pattern \(\Gate{H} \cdot \ket{0}\), the term \(\Gate{H}\) acts as a function which outputs a pattern, whereas in the expression \(\Gate{H}[x]\), the term \(\Gate{H}\) acts on a qubit \(x\) which it is then viewed as mutating, returning nothing. To remedy this, we add a sort of gates, which bind the free variables in a unitary, and allow them to be embedded into a pattern.

The syntax for this language has the following grammar, where \(\Gamma\) is a context and \(x\) is a variable:
\begin{alignat*}{4}
  &\mathsf{Copattern}\text{:}~~~~&\ &c&\ &\Coloneq&\ &x \mid c_1 \otimes c_2\\
  &\mathsf{Pattern}\text{:}&&p, q&&\Coloneq&&x \mid p \otimes q \mid g \cdot p \mid \ket{0} \mid \ket{1} \mid \ket{+} \mid \ket{-}\\
  &\mathsf{Unitary}\text{:}&&s, t&&\Coloneq&&\id{} \mid s;t \mid \Phase{\theta} \mid \ifthen{p = c}{s}\\
  &\mathsf{Gate}\text{:}&&g &&\Coloneq&&\Gamma \mapsto s
\end{alignat*}
As in the functional syntax, \(\mathsf{Copattern} \subset \mathsf{Pattern}\), but here the syntax for unitaries is separated.

We type this language with three judgements, one for pattern and copattern terms, one for unitary terms, and one for gates.
\[ \Gamma \vdash p : \qreg{n} \qquad \Gamma \vdash s \qquad \vdash g : \unitary{n}\]
The typing rules for patterns and copatterns are:
\vspace*{1mm}
\begin{mathpar}
  \inferrule{ }{x : \qreg{n} \vdash x : \qreg{n}}\and
  \inferrule{\Gamma \vdash s : \qreg{n} \and \Delta \vdash t : \qreg{m}}{\Gamma, \Delta \vdash s \otimes t : \qreg{(n + m)}}\and
  \inferrule{\vdash g : \unitary{n} \and \Gamma \vdash p : \qreg{n}}{\Gamma \vdash g \cdot p : \qreg{n}}
\end{mathpar}
\vspace*{2mm}
\begin{mathpar}
  \inferrule{ }{\cdot \vdash \ket{0} : \qreg{1}}\and
  \inferrule{ }{\cdot \vdash \ket{1} : \qreg{1}}\and
  \inferrule{ }{\cdot \vdash \ket{+} : \qreg{1}}\and
  \inferrule{ }{\cdot \vdash \ket{-} : \qreg{1}}\and
\end{mathpar}
\vspace*{1mm}
The typing rules for unitaries being:
\begin{mathpar}
  \inferrule{ }{\Gamma \vdash \id{}}\and
  \inferrule{\Gamma \vdash s \and \Gamma \vdash t}{\Gamma \vdash s; t}\and
  \inferrule{ }{\Gamma \vdash \Phase{\theta}}\and
  \inferrule{\Delta \vdash c : \qreg{m} \and \Theta \vdash p : \qreg{m} \and \Gamma, \Theta \vdash s}{\Gamma, \Delta \vdash \ifthen{p = c}{s}}
\end{mathpar}
Lastly, the unique typing rule for gates is the following, where \(\# \Gamma\) is the number obtained by summing the indices of each type in \(\Gamma\):
\begin{mathpar}
  \inferrule{\Gamma \vdash s \and \# \Gamma = n}{\vdash \Gamma \mapsto s : \unitary{n}}
\end{mathpar}
In addition, we allow typing rules which permute the context, much like in the functional version of the syntax, though note that due to the lack of return values in this syntax, there is no way to create a ``native'' controlled swap, (i.e. without inserting a swap gate using quantum operations).

We further note that the typing rule here for the ``if let'' expression allows variable capture, without any of the complications that arised in the functional syntax.

\section{Future Work}\label{s:conclusion}

We conclude by discussing some future directions for the development of languages using the constructs introduced in this work. An immediate direction is to extend the compilation and denotational semantics to the functional and imperative variants of the language presented in \cref{sec:beyond-combinators}.

In this paper we have focussed on using the ``if let'' construction as a device for writing quantum programs, yet it could also find use as a representation of quantum programs within an optimising compiler. At this higher level of abstraction than quantum circuits, more global optimisations become apparent. \Cref{s:compiler} already contains one instance of this, simplifying a conjugation appearing in a controlled block.

Additionally, a study of more comprehensive equational systems on our language could extend the semantic equivalences between terms presented in \cref{s:semantics}. This could take the form of a complete representation of equality between quantum programs (as has been done for quantum circuits~\cite{clementetal:equationalcircuits}), or aim to characterise an equality relation which may not be complete (with respect to the semantics in unitary matrices) but could exhibit an efficient normalisation function, with potential applications for optimisation.

Finally, the language presented here describes programs which are purely quantum, with no classical components, omitting even allocation and measurement. A more fully-featured quantum programming language based on the ``if let'' construction requires the addition of these operations. In contrast to naively extending the language, these operations can reuse parts of the pattern infrastructure, with measurement bases specified by patterns, and patterns specified by subprograms with allocation but without measurement.

\bibliographystyle{plain}
\bibliography{bibliography}

% \end{document}

\appendix

\section{Proof of Type-Soundness}\label{app:typing}

As stated in the main body, we prove \cref{thm:type-soundness} by mutually inducting on terms and patterns. We assume throughout that \((q, l, r) : k \to n\) is a (valid) evaluation context, with \(\vdash q : \pattern m n\) and \(l + k + r = m\).
\begin{itemize}
\item Case \(\Phase{\theta}\): In this case \(k = 0\) and so \(l + r = m\). Therefore, \(\vdash \Phase{\theta} \otimes \id{l + r} : \unitary m\), and so by the typing rule for ``if let'' we are done.
\item Case \(s;t\): Follows immediately from inductive hypothesis and the typing rule for composition.
\item Case \(s \otimes t\): By assumption, there are \(k_1, k_2\) such that \(k_1 + k_2 = k\) and \(\vdash s : \unitary{k_1}\) and \(\vdash t : \unitary {k_2}\). Therefore, \((q, l, r + k_2) : k_1 \to n\) and \((q, l + k_1, r) : k_2 \to n\) are evaluation contexts. By inductive hypothesis, the evaluations of \(s\) and \(t\) both have type \(\unitary{n}\), and so can be composed.
\item Case \(\id k\): By convention, the resulting normal term is \(\id n\), which has type \(\unitary{n}\).
\item Case \(\ifthen p s\): By the typing rule for ``if let'', we must have \(\vdash p : \pattern j k\) and \(\vdash s : \unitary{j}\). If \(\evalp {q,l,r} {p} = ([c_1, \dots, c_N], q')\), then by inductive hypothesis for patterns we have that \(\vdash c_1 ; \dots ; c_N : \unitary n\) (and also \(\vdash c_l^\dagger ; \dots ; c_1^\dagger : \unitary n\)) and that \((q', l, r) : j \to n\) is a valid evaluation context. Letting \(\evalu {q', l, r} {s} = [c'_1, \dots, c'_{l'}]\), by inductive hypothesis for terms we obtain \(\vdash c'_1 ; \dots ; c'_{l'} : \unitary{n}\). Therefore:
    \[ \vdash c_l^\dagger ; \dots ; c_1^\dagger ; c'_1 ; \dots ; c'_{l'} ; c_1 ; \dots ; c_l : \unitary n\]
  follows from the typing rule for composition.
\item Case \(\ket x\): Again by convention, the composition of the empty list is \(\id n\), which is trivially well-typed. It remains to show that \((q[\ket x / f(0)], l, r) : 0 \to n\) is an evaluation context, which requires \(l + r = m - 1\), but by hypothesis we have \(l + 1 + r = m\).
\item Case \(s\): Follows directly from the inductive hypothesis for terms.
\item Case \(p_1 \cdot p_2\): Suppose \(\vdash p_1 : \pattern j k\) and \(\vdash p_2 : \pattern i j\). Let \(\evalp {q,f} {p_1} = ([c_1, \dots, c_N], q')\). Then by inductive hypothesis we have that \((q', l, r) : j \to n\) is an evaluation context, and \(\vdash c_1 ; \dots ; c_N : \unitary{n}\). If \(\evalp {q', l, r} {p_2} = ([c'_1, \dots, c'_{N'}], q'')\), then by inductive hypothesis we have \((q'', l, r) : i \to n\) and \(\vdash c'_1 ; \dots ; c'_{N'} : \unitary{n}\), and so we are done by the typing rule for compositions.
\item Case \(p_1 \otimes p_2\): By assumption \(\vdash p_1 : \pattern {j_1} {k_1}\) and \(\vdash p_2 : \pattern{j_2} {k_2}\) such that \(k_1 + k_2 = k\). Then \(l + k_1 + k_2 + r = m\) and so we have that \((q, l, r + k_2) : k_1 \to n\) is an evaluation context. If \(\evalp {q, l, r + k_2}{p_1} = ([c_1, \dots, c_N], q')\), then by inductive hypothesis, \(\vdash c_1 ; \dots ; c_N : \unitary n\) and \((q', l, r + k_2) : j_1 \to n\) is an evaluation context. This implies that \(\vdash q' : \pattern {m'} {n}\) and \(l + j_1 + r + k_2 = m'\), and hence \((q', l + j_1, r) : k_2 \to n\) is also an evaluation context. Letting \(\evalp {q', l + j_1, r} {p_2} = ([c'_1, \dots, c'_{N'}], q'')\), by a second application of inductive hypothesis we have \(\vdash c'_1 ; \dots ; c'_{N'} : \unitary n\) and \((q'', l + j_1, r) : j_2 \to n\) is an evaluation context. By the typing rule for compositions we have:
  \[ \vdash c'_1 ; \dots ; c'_{N'} ; c_1 ; \dots ; c_N : \unitary n \]
and since \((q'', l + j_1, r) : j_2 \to n\), we have that \((q'', l, r) : j_1 + j_2 \to n\) is also an evaluation context, as required.
\end{itemize}

\section{Proofs of Section~\ref{s:semantics}} \label{app:semantics}

\begin{lemma}\label{lem:univ-monic}
	If an independent cospan $(f_1,f_2)$ is made up of dagger monomorphisms,
	then the mediating morphism $u$ obtained from an independent coproduct
	$(i_1, i_2)$ is also dagger monic.
\end{lemma}
\begin{proof}
	In this proof, we use the fact that an independent coproduct is jointly epic
	in two ways. If we have that $(i_1, i_2)$ is jointly epic, it means that
	\[
		\text{if we have }
		\left\{
			\begin{array}{l}
				g i_1 = h i_1 \\
				g i_2 = h i_2
			\end{array}
		\right.
		\text{, then }
		g = h;
	\]
	since we have a dagger category, the following holds:
	\[
		\text{if we have }
		\left\{
			\begin{array}{l}
			 	i_1\dg g' = i_1\dg h' \\
				i_2\dg g' = i_2\dg h'
			\end{array}
		\right.
		\text{, then }
		g' = h'.
	\]
	We can combine the two: if both $(i_1, i_2)$ and $(f'_1, f'_2)$ are jointly
	epic, it means that
	\[
		\text{if we have }
		\left\{
			\begin{array}{l}
				i_1\dg g f'_1 = i_1\dg h f'_1 \\
				i_2\dg g f'_1 = i_2\dg h f'_1 \\
				i_1\dg g f'_2 = i_1\dg h f'_2 \\
				i_2\dg g f'_2 = i_2\dg h f'_2
			\end{array}
		\right.
		\text{, then }
		\left\{
			\begin{array}{l}
				g f'_1 = h f'_1 \\
				g f'_2 = h f'_2
			\end{array}
		\right.
		\text{, then }
		g = h;
	\]
	we call this \emph{double joint epicness}.

	The following holds.
	\[
		\left\{
			\begin{array}{l}
				\iota_1\dg u\dg u \iota_1 = f_1\dg f_1 = \iid = \iota_1\dg \iota_1 \\
				\iota_1\dg u\dg u \iota_2 = f_1\dg f_2 = 0 = \iota_1\dg \iota_2 \\
				\iota_2\dg u\dg u \iota_1 = f_2\dg f_1 = 0 = \iota_2\dg \iota_1 \\
				\iota_2\dg u\dg u \iota_2 = f_2\dg f_2 = \iid = \iota_2\dg \iota_2 \\
			\end{array}
		\right.
	\]
	We conclude by double joint epicness that $u\dg u =
	\iid$, which is the definition of dagger monic.
\end{proof}

\begin{proof}[Proof of Lemma~\ref{lem:ic-choice}]
	The following holds.
	\[
		\left\{
			\begin{array}{l}
				p\dg cc\dg p = 0 = p\dg dd\dg p \\
				p\dg cc\dg c = 0 = p\dg dd\dg c \\
				d\dg cc\dg p = 0 = d\dg dd\dg p \\
				d\dg cc\dg c = d\dg c = d\dg dd\dg c
			\end{array}
		\right.
	\]
	We conclude by double joint epicness.
\end{proof}

\begin{lemma}\label{lem:med-composition}
	In a category with independent coproducts, if $(f,g)$ is an independent
	cospan and $h$ is a dagger monomorphism, then
	\[
		h \circ [f,g] = [hf,hg]
	\]
\end{lemma}
\begin{proof}
	The pair $(hf,hg)$ is indeed an independent cospan since $(hf)\dg hg = f\dg
	h\dg h g = f\dg g = 0$. Additionally, $h [f,g]$ is a mediating morphism,
	therefore we conclude by uniqueness.
\end{proof}

\subsection{Proof of Theorem~\ref{th:ic-smc}}

Let $\cat{C}$ be a dagger category with a zero object and independent coproduct
(meaning, a universal cospan of dagger monomorphisms that are jointly epic).
Fix an independent coproduct
\[
	A \stackrel{\iota_A}{\longrightarrow} A \oplus B \stackrel{\iota_B}{\longleftarrow} B
\]
for each pair of objects $A$ and $B$ of $\cat{C}$. We will prove that
$\big(\cat{C}, \oplus, O \big)$ is a dagger symmetric monoidal category.
%\todo{Cite [Simpson, ``Category-theoretic Structure for
%Independence and Conditional Independence'', Theorem 3.9].}

\begin{proposition}\label{bifunctoriality}
	There is a dagger functor $\oplus \colon \cat{C} \times \cat{C} \to \cat{C}$.
\end{proposition}
\begin{proof}
	Let $f \colon A \to A'$ and $g \colon B \to B'$ be morphisms.
	Define $f \oplus g \colon A \oplus B \to A' \oplus B'$ to be the unique
	morphism satisfying $(f \oplus g) \iota_A = \iota_{A'} f$ and $(f \oplus g)
	\iota_B = \iota_{B'} g$.
	The uniqueness immediately shows that $1_A \oplus 1_B = 1_{A \oplus B}$.
	Let $f' \colon A' \to A''$ and $g \colon B' \to B''$ be additional
	morphisms. It again follows directly from the uniqueness property that $(f'
	\oplus g') (f \oplus g) = (f'f) \oplus (g'g)$.

	Additionally, we have:
	\begin{align*}
		\iota_A \dg (f\dg \oplus g\dg) \iota_{A'}
		&= \iota_A\dg \iota_A f\dg
		= f\dg = f\dg \iota_{A'}\dg \iota_{A'} \\
		&= (\iota_{A'} f)\dg \iota_{A'}
		= ((f \oplus g) \iota_A)\dg \iota_{A'}
		= \iota_A\dg (f \oplus g)\dg \iota_{A'} \\
		\iota_A \dg (f\dg \oplus g\dg) \iota_{B'}
		&= \iota_A\dg \iota_B g\dg
		= 0 = f\dg \iota_{A'}\dg \iota_{B'} \\
		&= (\iota_{A'} f)\dg \iota_{B'}
		= ((f \oplus g) \iota_A)\dg \iota_{B'}
		= \iota_A\dg (f \oplus g)\dg \iota_{B'} \\
		\iota_B \dg (f\dg \oplus g\dg) \iota_{A'}
		&= \iota_B\dg \iota_A f\dg
		= 0 = g\dg \iota_{B'}\dg \iota_{A'} \\
		&= (\iota_{B'} g)\dg \iota_{A'}
		= ((f \oplus g) \iota_B)\dg \iota_{A'}
		= \iota_B\dg (f \oplus g)\dg \iota_{A'} \\
		\iota_B \dg (f\dg \oplus g\dg) \iota_{B'}
		&= \iota_B\dg \iota_B g\dg
		= g\dg = g\dg \iota_{B'}\dg \iota_{B'} \\
		&= (\iota_{B'} g)\dg \iota_{B'}
		= ((f \oplus g) \iota_B)\dg \iota_{B'}
		= \iota_B\dg (f \oplus g)\dg \iota_{B'}
	\end{align*}
	by double joint epicness, $(f \oplus g)\dg = f\dg \oplus g\dg$.
\end{proof}

\begin{lemma}\label{unitors}
	There are natural dagger isomorphisms $\lambda_A \colon O \oplus A \to A$
	and $\rho_A \colon A \oplus O \to A$.
\end{lemma}
\begin{proof}
	Write $(O \oplus A,0,\iota)$ for the chosen independent coproduct of $0$ and $A$.
	Since $(A,0,\iid)$ is another cospan of $0$ and $A$, there is a unique
	dagger monomorphism (see Lemma~\ref{lem:univ-monic}) $\lambda_A \colon O
	\oplus A \to A$ satisfying $\lambda \iota = \iid$. Now $\lambda = \lambda
	\iota\iota^\dag = \iota^\dag$, so $\iota^\dag \iota = \lambda \iota = \iid$.
	Therefore $\lambda=\iota^\dag$ is a dagger isomorphism.
	\[\begin{tikzcd}[ampersand replacement=\&, column sep=2cm, row sep=1cm]
		O \& {O \oplus A} \& A \\
		O \& {O \oplus B} \& B
		\arrow["\iid"', from=1-1, to=2-1]
		\arrow["0"', from=1-1, to=1-2]
		\arrow["{\lambda_A}"', shift right, from=1-2, to=1-3]
		\arrow["{\iid \oplus f}"{description}, dashed, from=1-2, to=2-2]
		\arrow["{\iota_A}"', shift right, from=1-3, to=1-2]
		\arrow["f", from=1-3, to=2-3]
		\arrow["0", from=2-1, to=2-2]
		\arrow["{\lambda_B}", shift left, from=2-2, to=2-3]
		\arrow["{\iota_B}", shift left, from=2-3, to=2-2]
	\end{tikzcd}\]

	If $f \colon A \to B$, then by definition $(\iid \oplus f) \iota_A = \iota_B
	f$. Therefore
	\[
		\lambda_B (\iid \oplus f)
		= \lambda_B (\iid \oplus f) \iota_A \lambda_A
		= \lambda_B \iota_B f \lambda_A
		= f \lambda_A \text.
	\]
	Thus $\lambda$ is natural.
	The proof for $\rho$ is similar.
\end{proof}

\begin{lemma}\label{symmetry}
	There are natural dagger isomorphisms $\sigma_{A,B} \colon A \oplus B \to B
	\oplus A$ satisfying $\sigma_{B,A} \sigma_{A,B} = 1$.
\end{lemma}
\begin{proof}
	If $(A \oplus B,i_1,i_2)$ and $(B \oplus A,j_1,j_2)$ are the chosen
	independent coproducts, then $(B \oplus A, j_2, j_1)$ is an independent
	cospan as well, so there is a unique dagger monomorphism
	(see Lemma~\ref{lem:univ-monic}) $\sigma_{A,B} \colon A \oplus B \to B \oplus
	A$ satisfying $j_2 = \sigma i_1$ and $j_1 = \sigma i_2$. It follows from
	uniqueness of mediating morphisms between independent coproducts that
	$\sigma_{B,A} \sigma_{A,B} = \iid$.
	\[\begin{tikzcd}[ampersand replacement=\&]
		A \&\& B \&\& A \\
		\& {A \oplus B} \&\& {B \oplus A} \\
		\& {A' \oplus B'} \&\& {B' \oplus A'} \\
		{A'} \&\& {B'} \&\& {A'}
		\arrow["\iid"{description}, curve={height=-18pt}, from=1-1, to=1-5]
		\arrow["f"', from=1-1, to=4-1]
		\arrow["g"{description}, from=1-3, to=4-3]
		\arrow["f", from=1-5, to=4-5]
		\arrow["{i_1}"', from=1-1, to=2-2]
		\arrow["{i_2}", from=1-3, to=2-2]
		\arrow["{\sigma_{A,B}}"{description}, from=2-2, to=2-4]
		\arrow["{f \oplus g}"{description}, from=2-2, to=3-2]
		\arrow["{j_1}"', from=1-3, to=2-4]
		\arrow["{j_2}", from=1-5, to=2-4]
		\arrow["{g \oplus f}"{description}, from=2-4, to=3-4]
		\arrow["{\sigma_{A',B'}}"{description}, from=3-2, to=3-4]
		\arrow["{i'_1}", from=4-1, to=3-2]
		\arrow["{i'_2}"', from=4-3, to=3-2]
		\arrow["{j'_1}", from=4-3, to=3-4]
		\arrow["{j'_2}"', from=4-5, to=3-4]
		\arrow["\iid"{description}, curve={height=18pt}, from=4-1, to=4-5]
	\end{tikzcd}\]
	If we have $f \colon A \to A'$ and $g \colon B \to B'$, and we write $(A'
	\oplus B',i'_1,i'_2)$ and $(B' \oplus A',j'_1,j'_2)$ for the chosen
	independent coproducts, then
	\begin{align*}
		(g \oplus f) \sigma_{A,B} i_1
		&= (g \oplus f) j_2
		= j'_2 f
		= \sigma_{A',B'} i'_1 f
		= \sigma_{A',B'} (f \oplus g) i_1 \text, \\
		(g \oplus f) \sigma_{A,B} i_2
		&= (g \oplus f) j_1
		= j'_1 g
		= \sigma_{A',B'} i'_2 g
		= \sigma_{A',B'} (f \oplus g) i_2 \text.
	\end{align*}
	It now follows from joint epicness of $(i_1, i_2)$ that $(g \oplus f)
	\sigma_{A,B}= \sigma_{B',A'}(f \oplus g)$. Thus $\sigma$ is natural.
\end{proof}

\begin{lemma}\label{associators}
	There are natural dagger isomorphisms $\alpha_{A,B,C} \colon A \oplus (B
	\oplus C) \to (A \oplus B) \oplus C$.
\end{lemma}
\begin{proof}
	Write $(B \oplus C,j_1,j_2)$, $\big(A \oplus (B \oplus C), i_1,i_2\big)$,
	$(A \oplus B, i'_1,i'_2)$, and $\big((A \oplus B) \oplus C, j'_1,j'_2)$ for
	the chosen independent coproducts.
	Because $\big( A \oplus (B \oplus C), i_1, i_2 j_1 \big)$ is an independent
	cospan of $A$ and $B$, there is a unique dagger monomorphism $m \colon A
	\oplus B \to A \oplus (B \oplus C)$ satisfying $m i'_1 = i_1$ and $m i'_2
	= i_2 j_1$. Similarly, there is a unique dagger monomorphism
	$\beta_{A,B,C} \colon (A \oplus B) \oplus C \to A \oplus (B \oplus C)$
	satisfying $\beta j'_1 = m$ and $\beta j'_2 = i_2 j_2$. In the same way,
	there are unique dagger monomorphisms $n \colon B \oplus C \to (A \oplus B)
	\oplus C$ and $\alpha \colon A \oplus (B \oplus C) \to (A \oplus B) \oplus
	C$ satisfying $n j_1 = j'_1 i'_2$, $n j_2 = j'_2$, $\alpha i_1 = j'_1 i'_1$
	and $\alpha i_2 = n$.
	Now
	\begin{align*}
		\beta \alpha i_1 &= \beta j'_1 i'_1 = m i'_1 = i_1 \text, \\
		\beta \alpha i_2 j_1 &= \beta n j_1 = \beta j'_1 i'_2 = m i'_2 = i_2 j_1 \text, \\
		\beta \alpha i_2 j_2 &= \beta n j_2 = \beta j'_2 = i_2 j_2 \text.
	\end{align*}
	since $(j_1, j_2)$ is jointly epic, the second and third line about give
	$\beta \alpha i_2 = i_2$ and because $(i_1, i_2)$ is jointly epic, we have
	$\beta \alpha = \iid$. It follows that $\alpha^\dag = \beta \alpha
	\alpha^\dag = \beta$, and that $\alpha$ is dagger isomorphism.

	
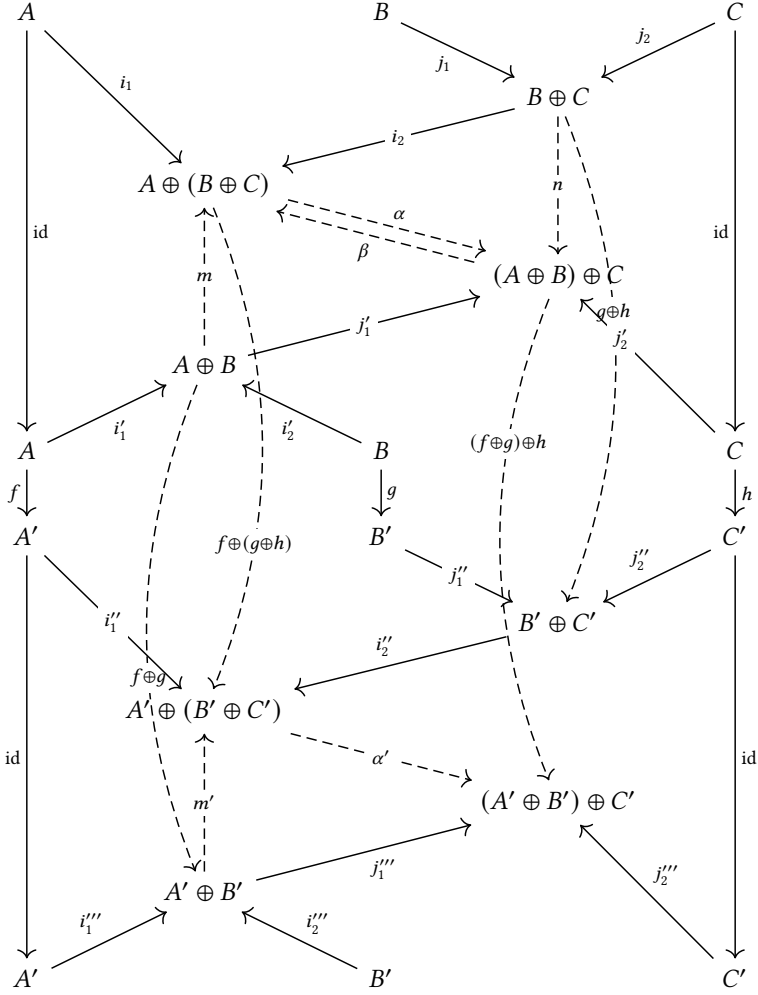
\[\begin{tikzcd}[ampersand replacement=\&, xscale=2]
		A \&\& B \&\& C \\
		\&\&\& {B \oplus C} \\
		\& {A \oplus (B \oplus C)} \\
		\&\&\& {(A \oplus B) \oplus C} \\
		\& {A \oplus B} \\
		A \&\& B \&\& C \\
		{A'} \&\& {B'} \&\& {C'} \\
		\&\&\& {B' \oplus C'} \\
		\& {A' \oplus (B' \oplus C')} \\
		\&\&\& {(A' \oplus B') \oplus C'} \\
		\& {A' \oplus B'} \\
		{A'} \&\& {B'} \&\& {C'}
		\arrow["\iid", from=1-1, to=6-1]
		\arrow["\iid"', from=1-5, to=6-5]
		\arrow["{j_1}"', to=2-4, from=1-3]
		\arrow["{j_2}"', to=2-4, from=1-5]
		\arrow["{g \oplus h}"{description, pos=0.4}, curve={height=-14pt}, dashed, from=2-4, to=8-4]
		\arrow["{i_1}", to=3-2, from=1-1]
		\arrow["{i_2}"{description}, to=3-2, from=2-4]
		\arrow["\alpha", shift left, dashed, from=3-2, to=4-4]
		\arrow["m"{description}, dashed, to=3-2, from=5-2]
		\arrow["{f \oplus (g \oplus h)}"{description, pos=0.7}, curve={height=-14pt}, dashed, from=3-2, to=9-2]
		\arrow["n"{description}, dashed, to=4-4, from=2-4]
		\arrow["\beta", shift left, dashed, from=4-4, to=3-2]
		\arrow["{j'_1}"{description}, to=4-4, from=5-2]
		\arrow["{j'_2}"{description, pos=0.7}, to=4-4, from=6-5]
		\arrow["{(f \oplus g) \oplus h}"{description, pos=0.3}, curve={height=14pt}, dashed, from=4-4, to=10-4]
		\arrow["{i'_1}"', to=5-2, from=6-1]
		\arrow["{i'_2}", to=5-2, from=6-3]
		\arrow["{f \oplus g}"{description, pos=0.6}, curve={height=14pt}, dashed, from=5-2, to=11-2]
		\arrow["f"', from=6-1, to=7-1]
		\arrow["g", from=6-3, to=7-3]
		\arrow["h", from=6-5, to=7-5]
		\arrow["\iid"', from=7-1, to=12-1]
		\arrow["\iid", from=7-5, to=12-5]
		\arrow["{j''_1}"{description}, to=8-4, from=7-3]
		\arrow["{j''_2}"', to=8-4, from=7-5]
		\arrow["{i''_1}"{description}, to=9-2, from=7-1]
		\arrow["{i''_2}"', to=9-2, from=8-4]
		\arrow["{\alpha'}"{description}, dashed, from=9-2, to=10-4]
		\arrow["{m'}"{description}, dashed, to=9-2, from=11-2]
		\arrow["{j'''_1}"', to=10-4, from=11-2]
		\arrow["{j'''_2}"', to=10-4, from=12-5]
		\arrow["{i'''_1}", to=11-2, from=12-1]
		\arrow["{i'''_2}"', to=11-2, from=12-3]
	\end{tikzcd}\]

	Let $f \colon A \to A'$, $g \colon B \to B'$, and $h \colon C \to C'$ be
	morphisms, and write the chosen independent coproducts of $A'$,
	$B'$, and $C'$ as above. Then
	\begin{align*}
		((f \oplus g)\oplus h) \alpha i_2 j_1
		&= ((f \oplus g)\oplus h) n j_1
		= ((f \oplus g)\oplus h) j'_1 i'_2
		= j'''_1 (f \oplus g) i'_2
		= j'''_1 i'''_2 g \\
		&= \alpha' m' i'''_2 g
		= \alpha' i''_2 j''_1 g
		= \alpha' i''_2 (g \oplus h) j_1
		= \alpha' (f \oplus (g \oplus h)) i_2 j_1 \text, \\
		((f \oplus g)\oplus h) \alpha i_2 j_2
		&= ((f \oplus g)\oplus h) n j_2
		= ((f \oplus g)\oplus h) j'_2
		= j'''_2 h \\
		&= \alpha' i''_2 j''_2 h
		= \alpha' i''_2 (g \oplus h) j_2
		= \alpha' (f \oplus (g \oplus h)) i_2 j_2 \text.
	\end{align*}
	Because $(j_1, j_2)$ is jointly epic, it implies that
	\[
		((f \oplus g)\oplus h) \alpha i_2 = \alpha' (f \oplus (g \oplus h)) i_2 \text.
	\]
	Additionally
	\begin{align*}
		((f \oplus g)\oplus h) \alpha i_1
		&= ((f \oplus g)\oplus h) j'_1 i'_1
		= j'''_1 (f \oplus g) i'_1
		= j'''_1 i'''_1 f \\
		&= \alpha' m' i'''_1 f
		= \alpha' i''_1 f
		= \alpha' (f \oplus (g \oplus h)) i_1 \text.
	\end{align*}
	Since $(i_1, i_2)$ is jointly epic, it implies that $((f \oplus g) \oplus h)
	\alpha = \alpha' (f \oplus (g \oplus h))$. That is, $\alpha$ is natural.
\end{proof}

Note that, since $\cat C$ is a dagger category, independent coproduct are
actually independent \emph{biproducts}. Indeed, let $(i_1, i_2)$ be an
independent coproduct and $(f_1, f_2)$ be a \emph{span} this time:
\[
	\begin{tikzcd}
		A & C & B
		\arrow["f_1"', from=1-2, to=1-1]
		\arrow["f_2", from=1-2, to=1-3]
	\end{tikzcd}
\]
We know that there is a unique $u$ such that the following diagram commutes:
\[
	\begin{tikzcd}[column sep=large, row sep=large]
		A & A \oplus B & B \\
		& C
		\arrow["i_1", from=1-1, to=1-2]
		\arrow["i_2"', from=1-3, to=1-2]
		\arrow["f_1\dg"', from=1-1, to=2-2]
		\arrow["f_2\dg", from=1-3, to=2-2]
		\arrow["{u}"{description}, dashed, to=2-2, from=1-2]
	\end{tikzcd}
\]
which implies that $u\dg$ is the unique morphism such that the following diagram commutes:
\[
	\begin{tikzcd}[column sep=large, row sep=large]
		A & A \oplus B & B \\
		& C
		\arrow["i_1\dg"', from=1-2, to=1-1]
		\arrow["i_2\dg", from=1-2, to=1-3]
		\arrow["f_1", from=2-2, to=1-1]
		\arrow["f_2"', from=2-2, to=1-3]
		\arrow["{u\dg}"{description}, dashed, to=1-2, from=2-2]
	\end{tikzcd}
\]
thus making $(i_1\dg, i_2\dg)$ an independent \emph{product}. In other words,
the description of an independent coproduct is equivalent to a description of
an independent product.

\begin{proposition}
	If $\cat{C}$ is a dagger category, then any choice of zero object $O$ and
	independent (co)products $A \oplus B$ makes $\cat{C}$ into a dagger
	symmetric monoidal category.
\end{proposition}
\begin{proof}
	Combine \Cref{bifunctoriality,unitors,symmetry,associators} with the
	following. Write $p^{A,B}_A \colon A \oplus B \to A$ and $p^{A,B}_B \colon
	A \oplus B \to B$ for the chosen independent products.
	The triangle, pentagon, and hexagon equations can be proven algebraically
	exactly as in the case of cartesian
	categories~\cite[Example~5.3.6]{yanofsky:monoidalcategories}. We show the
	triangle equation as an example.
	\[\begin{tikzcd}[ampersand replacement=\&, row sep=5mm]
		A \&\& O \&\& B \\
		\&\&\& {O \oplus B} \\
		\& {A \oplus (O \oplus B)} \\
		\&\&\& {(A \oplus O) \oplus B} \\
		\& {A \oplus O} \\
		A \&\&\&\& B \\
		\& O \\
		A \&\& {A \oplus B} \&\& B
		\arrow["\iid"', from=1-1, to=6-1]
		\arrow["\iid", from=1-5, to=6-5]
		\arrow["0"', from=2-4, to=1-3]
		\arrow["\lambda", from=2-4, to=1-5]
		\arrow["p^{A,O \oplus B}_A"{description}, from=3-2, to=1-1]
		\arrow["p^{A,O\oplus B}_{O \oplus B}"{description}, from=3-2, to=2-4]
		\arrow["\alpha"', dashed, from=3-2, to=4-4]
		\arrow["m"', dashed, from=3-2, to=5-2]
		\arrow["{1 \oplus \lambda}"{pos=0.6}, dashed, from=3-2, to=8-3]
		\arrow["p^{A \oplus O,B}_{A \oplus O}"{description}, from=4-4, to=5-2]
		\arrow["p^{A \oplus O,B}_B"{description}, from=4-4, to=6-5]
		\arrow["\rho \oplus 1", dashed, from=4-4, to=8-3]
		\arrow["\rho"', from=5-2, to=6-1]
		\arrow["0", from=5-2, to=7-2]
		\arrow["\iid"', from=6-1, to=8-1]
		\arrow["\iid"', from=6-5, to=8-5]
		\arrow["0", from=7-2, to=8-1]
		\arrow["{p^{A,B}_A}", from=8-3, to=8-1]
		\arrow["{p^{A,B}_B}"', from=8-3, to=8-5]
	\end{tikzcd}\]
	In the above diagram:
	\begin{align*}
		p^{A,B}_A (\rho \oplus \iid) \alpha
		& = \rho p^{A \oplus O,B}_{A \oplus O} \alpha
		= \rho m
		= p^{A,O \oplus B}_A
		= p^{A,B}_A (\iid \oplus \lambda)\text, \\
		p^{A,B}_B (\rho \oplus \iid) \alpha
		& = p^{A \oplus O,B}_B \alpha
		= \lambda p^{A,O \oplus B}_{O \oplus B}
		= p^{A \oplus B}_B (\iid \oplus \lambda)\text.
	\end{align*}
	By the universal property of $A \oplus B$, therefore $(\rho \oplus \iid)
	\alpha = \iid \oplus \lambda$.
\end{proof}

\subsection{Proof of Theorem~\ref{th:rig}}

\newcommand{\stimes}{{\otimes}}
\newcommand{\splus}{{\oplus}}

Let $(\cat C, \otimes, I)$ be a dagger symmetric monoidal category with a zero
object $O$, such that there is a dagger isomorphism $O \otimes A \cong A$ for
all $A$, and independent coproducts, preserved by the tensor $\otimes$. With
Theorem~\ref{th:ic-smc}, we know that $(\cat C, \oplus, O)$ is a dagger
symmetric monoidal category.

We will prove that $(\cat C, \oplus, O, \otimes, I)$ is a dagger rig
category.

First, we can define the distributor as the following mediating morphism.
\[
	\begin{tikzcd}[ampersand replacement=\&, row sep=large, column sep=2cm]
		A \otimes C \& (A \otimes C) \oplus (B \otimes C) \& B \otimes C \\
		\& (A \oplus B) \otimes C
		\arrow["i_1"', from=1-1, to=1-2]
		\arrow["0"', bend left=20, from=1-1, to=1-3]
		\arrow["i_1 \otimes \iid_C"', from=1-1, to=2-2]
		\arrow["d"{description},dashed, from=1-2, to=2-2]
		\arrow["i_2", from=1-3, to=1-2]
		\arrow["i_2 \otimes \iid_C", from=1-3, to=2-2]
	\end{tikzcd}
\]
Since both $(i_1, i_2)$ and $(i_1 \otimes \iid, i_2 \otimes \iid)$ are
indenpendent coproducts, $d$ is necessarily a dagger isomorphism.

The rest of the proof consists in verifying all the 23 coherence
diagrams~\cite[Section 2.1]{bimonoidal-book} that define a rig category. We
pick some of them to display the proof strategy.

\begin{lemma}
	The following diagram commutes for all objects $A, B, C$.
	\[
		\begin{tikzcd}[ampersand replacement=\&, row sep=large, column sep=2cm]
			(A \otimes C) \oplus (B \otimes C) \& (A \oplus B) \otimes C \\
			(B \otimes C) \oplus (A \otimes C) \& (B \oplus A) \otimes C 
			\arrow["d", from=1-1, to=1-2]
			\arrow["d", from=2-1, to=2-2]
			\arrow["\sigma", from=1-1, to=2-1]
			\arrow["\sigma \otimes \iid", from=1-2, to=2-2]
		\end{tikzcd}
	\]
\end{lemma}
\begin{proof}
	By definition, the following diagram commutes:
	\[
		\begin{tikzcd}[row sep=large, column sep=large]
			{A \otimes C} &&&& {B \otimes C} \\
			& {(A \otimes C) \oplus (B \otimes C)} && {(A \oplus B) \otimes C} \\
			\\
			& {(B \otimes C) \oplus (A \otimes C)} && {(B \oplus A) \otimes C} \\
			{B \otimes C} &&&& {A \otimes C}
			\arrow["{i_1}"{description}, from=1-1, to=2-2]
			\arrow["{i_1 \otimes C}"{description}, from=1-1, to=2-4]
			\arrow["{i_2}"{description}, from=1-5, to=2-2]
			\arrow["{i_2 \otimes C}"{description}, from=1-5, to=2-4]
			\arrow["d"{description}, dashed, from=2-2, to=2-4]
			\arrow["\sigma"{description}, dashed, from=2-2, to=4-2]
			\arrow["{\sigma \otimes C}"{description}, dashed, from=2-4, to=4-4]
			\arrow["d'"{description}, dashed, from=4-2, to=4-4]
			\arrow["{i'_1}"{description}, from=5-1, to=4-2]
			\arrow["{i'_1 \otimes C}"{description}, from=5-1, to=4-4]
			\arrow["{i'_2}"{description}, from=5-5, to=4-2]
			\arrow["{i'_2 \otimes C}"{description}, from=5-5, to=4-4]
		\end{tikzcd}
	\]
	We have then
	\begin{align*}
		(\sigma \otimes \iid) d i_1
		&= (\sigma \otimes \iid) (i_1 \otimes \iid)
		= (i'_2 \otimes \iid)
		= d' i'_2
		= d' \sigma i_1\text, \\
		(\sigma \otimes \iid) d i_2
		&= (\sigma \otimes \iid) (i_2 \otimes \iid)
		= (i'_1 \otimes \iid)
		= d' i'_1
		= d' \sigma i_2 \text.
	\end{align*}
	By joint epicness, $(\sigma \otimes \iid) d = d' \sigma$.
\end{proof}

\begin{lemma}
	The following diagram commutes for all objects $A, B, C, D$.
	\[
		\begin{tikzcd}[ampersand replacement=\&, row sep=large, column sep=2cm]
			(A \otimes D) \oplus ((B \otimes D) \oplus (C \otimes D)) \& ((A
			\otimes D) \oplus (B \otimes D)) \oplus (C \otimes D) \\
			(A \otimes D) \oplus ((B \oplus C) \otimes D)) \& ((A
			\oplus B) \otimes D) \oplus (C \otimes D) \\
			(A \oplus (B \oplus C)) \otimes D \& ((A \oplus B) \oplus C) \otimes D 
			\arrow["\alpha^\oplus", from=1-1, to=1-2]
			\arrow["\alpha^\oplus \otimes \iid", from=3-1, to=3-2]
			\arrow["\iid \oplus d", from=1-1, to=2-1]
			\arrow["d", from=2-1, to=3-1]
			\arrow["d \oplus \iid", from=1-2, to=2-2]
			\arrow["d", from=2-2, to=3-2]
		\end{tikzcd}
	\]
\end{lemma}
\begin{proof}
	By definition, the following diagram commutes:
	
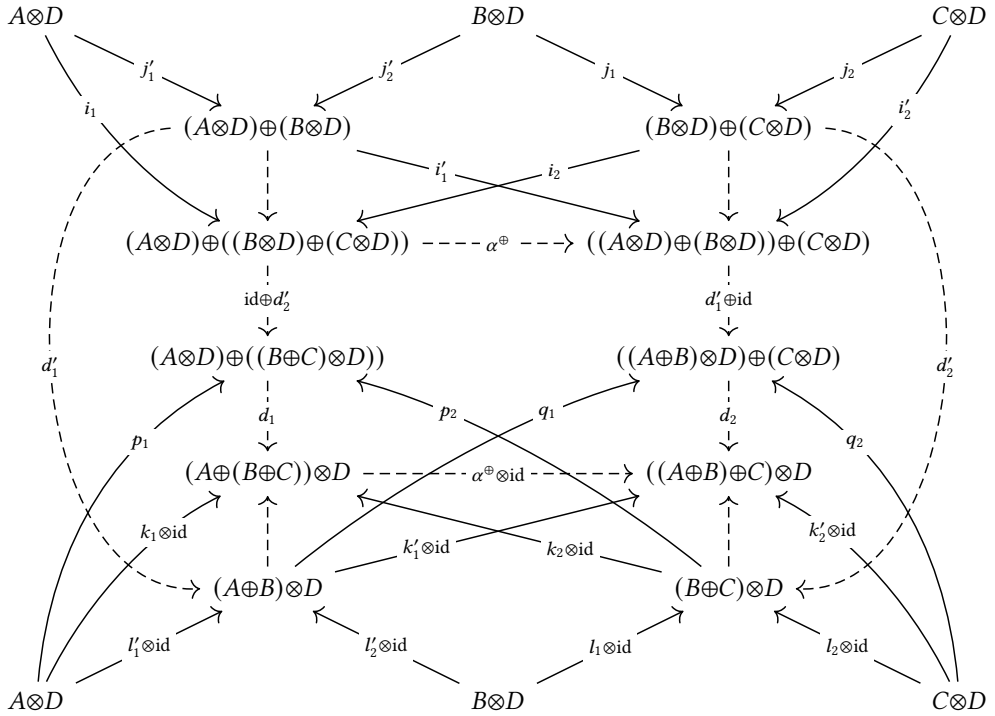
\[
		\begin{tikzcd}[ampersand replacement=\&, row sep=large, column sep=0.5cm]
			A \stimes D \&\& B \stimes D \&\& C \stimes D \\
			\& (A \stimes D) \splus (B \stimes D) \&\& (B \stimes D) \splus (C \stimes D) \& \\
			\& (A \stimes D) \splus ((B \stimes D) \splus (C \stimes D)) \&\& ((A
			\stimes D) \splus (B \stimes D)) \splus (C \stimes D) \& \\
			\& (A \stimes D) \splus ((B \splus C) \stimes D)) \&\& ((A
			\splus B) \stimes D) \splus (C \stimes D) \& \\
			\& (A \splus (B \splus C)) \stimes D \&\& ((A \splus B) \splus C) \stimes D \& \\
			\& (A \splus B) \stimes D \&\& (B \splus C) \stimes D \& \\
			A \stimes D \&\& B \stimes D \&\& C \stimes D 
			\arrow["{j_1}"{description}, from=1-3, to=2-4]
			\arrow["{j_2}"{description}, from=1-5, to=2-4]
			\arrow["{i_1}"{description, pos=0.3}, curve={height=15pt}, from=1-1, to=3-2]
			\arrow["{i_2}"{description,pos=0.3}, from=2-4, to=3-2]
			\arrow["{j'_1}"{description}, from=1-1, to=2-2]
			\arrow["{j'_2}"{description}, from=1-3, to=2-2]
			\arrow["{i'_1}"{description, pos=0.3}, from=2-2, to=3-4]
			\arrow["{i'_2}"{description,pos=0.3}, curve={height=-15pt}, from=1-5, to=3-4]
			\arrow["{\alpha^\splus}"{description}, dashed, from=3-2, to=3-4]
			\arrow[dashed, from=2-2, to=3-2]
			\arrow[dashed, from=2-4, to=3-4]
			\arrow["{l_1 \stimes\iid}"{description}, from=7-3, to=6-4]
			\arrow["{l_2 \stimes\iid}"{description}, from=7-5, to=6-4]
			\arrow["{k_1 \stimes\iid}"{description,pos=0.7}, curve={height=-15pt}, from=7-1, to=5-2]
			\arrow["{k_2 \stimes\iid}"{description,pos=0.3}, from=6-4, to=5-2]
			\arrow["{l'_1 \stimes\iid}"{description}, from=7-1, to=6-2]
			\arrow["{l'_2 \stimes\iid}"{description}, from=7-3, to=6-2]
			\arrow["{k'_1 \stimes\iid}"{description, pos=0.3}, from=6-2, to=5-4]
			\arrow["{k'_2 \stimes\iid}"{description,pos=0.7}, curve={height=15pt}, from=7-5, to=5-4]
			\arrow[dashed, from=6-2, to=5-2]
			\arrow[dashed, from=6-4, to=5-4]
 			\arrow["{\alpha^\splus \stimes \iid}"{description}, dashed, from=5-2, to=5-4]
			\arrow["{p_1}"{description,pos=0.7}, curve={height=-30pt}, from=7-1, to=4-2]
			\arrow["{p_2}"{description,pos=0.7}, curve={height=15pt}, from=6-4, to=4-2]
			\arrow["{q_1}"{description,pos=0.7}, curve={height=-15pt}, from=6-2, to=4-4]
			\arrow["{q_2}"{description,pos=0.7}, curve={height=30pt}, from=7-5, to=4-4]
 			\arrow["{d'_1}"{description}, bend right=90, dashed, from=2-2, to=6-2]
 			\arrow["{d'_2}"{description}, bend left=90, dashed, from=2-4, to=6-4]
 			\arrow["{\iid \splus d'_2}"{description}, dashed, from=3-2, to=4-2]
 			\arrow["{d_1}"{description}, dashed, from=4-2, to=5-2]
 			\arrow["{d'_1 \splus \iid}"{description}, dashed, from=3-4, to=4-4]
 			\arrow["{d_2}"{description}, dashed, from=4-4, to=5-4]
%  			\arrow["{\iid}"{description}, shift right=2, from=1-1, to=7-1]
%  			\arrow["{\iid}"{description}, shift left=2, from=1-5, to=7-5]
		\end{tikzcd}
	\]
	and we conclude by successive joint epicness of $(j_1,j_2)$ then of $(i_1,i_2)$.
\end{proof}

\begin{lemma}
	The following diagram commutes for all objects $A, B, C, D$.
	\[
		\begin{tikzcd}[ampersand replacement=\&, row sep=large, column sep=2cm]
			((A \otimes B) \otimes C) \oplus ((A \otimes B) \otimes D)
			\&
			(A \otimes (B \otimes C)) \oplus (A \otimes (B \otimes D))
			\\
			\&
			A \otimes ((B \otimes C) \oplus (B \otimes D))
			\\
			(A \otimes B) \otimes (C \oplus D)
			\&
			A \otimes (B \otimes (C \oplus D))
 			\arrow["\alpha^\otimes \oplus \alpha^\otimes", from=1-1, to=1-2]
 			\arrow["\alpha^\otimes", from=3-1, to=3-2]
 			\arrow["d", from=1-1, to=3-1]
 			\arrow["d", from=1-2, to=2-2]
 			\arrow["\iid \otimes d", from=2-2, to=3-2]
		\end{tikzcd}
	\]
\end{lemma}
\begin{proof}
	By definition, the following diagram commutes:
	\[
		\begin{tikzcd}[ampersand replacement=\&, row sep=large, column sep=1.3cm]
			(A \stimes B) \stimes C
			\&
			A \stimes (B \stimes C)
			\\
			(A \stimes B) \stimes D
			\&
			A \stimes (B \stimes D)
			\\
			((A \stimes B) \stimes C) \splus ((A \stimes B) \stimes D)
			\&
			(A \stimes (B \stimes C)) \splus (A \stimes (B \stimes D))
			\\
			\&
			A \stimes ((B \stimes C) \splus (B \stimes D))
			\\
			(A \stimes B) \stimes (C \splus D)
			\&
			A \stimes (B \stimes (C \splus D))
			\\
			(A \stimes B) \stimes C
			\&
			A \stimes (B \stimes C)
			\\
			(A \stimes B) \stimes D
			\&
			A \stimes (B \stimes D)
  			\arrow["\alpha^\stimes", from=1-1, to=1-2]
  			\arrow["\alpha^\stimes", from=2-1, to=2-2]
			\arrow["{j_1}"{description, pos=0.3}, curve={height=50pt}, from=1-1, to=3-1]
			\arrow["{j_2}"{description}, from=2-1, to=3-1]
			\arrow["{i_1}"{description, pos=0.3}, curve={height=50pt}, from=1-2, to=3-2]
			\arrow["{i_2}"{description}, from=2-2, to=3-2]
  			\arrow["\alpha^\stimes \splus \alpha^\stimes", from=3-1, to=3-2]
			\arrow["{\iid \stimes k_2}"{description, pos=0.3}, curve={height=-50pt}, from=7-1, to=5-1]
			\arrow["{\iid \stimes k_1}"{description}, from=6-1, to=5-1]
			\arrow["{\iid \stimes (\iid \stimes k_2)}"{description, pos=0.3}, curve={height=-50pt}, from=7-2, to=5-2]
			\arrow["{\iid \stimes (\iid \stimes k_1)}"{description}, from=6-2, to=5-2]
  			\arrow["\alpha^\stimes", from=5-1, to=5-2]
  			\arrow["\alpha^\stimes", from=6-1, to=6-2]
  			\arrow["\alpha^\stimes", from=7-1, to=7-2]
			\arrow["{\iid \stimes i'_2}"{description}, bend right=90, from=7-2, to=4-2]
			\arrow["{\iid \stimes i'_1}"{description}, curve={height=60pt}, from=6-2, to=4-2]
			\arrow["{d_1}"{description}, dashed, from=3-1, to=5-1]
			\arrow["{d_2}"{description}, dashed, from=3-2, to=4-2]
			\arrow["{\iid \stimes d'_2}"{description}, dashed, from=4-2, to=5-2]
		\end{tikzcd}
	\]
	and we conclude by joint epicness of $(j_1,j_2)$.
\end{proof}

\subsection{Proof of Theorem~\ref{th:well-defined-sem}}

We wish to prove that for all well-typed patterns $p$ and unitaries $s$, we have that:
\begin{itemize}
	\item there exists a morphism $\sem p ^\perp$ such that $\left( \sem p,
		\sem p ^\perp \right)$ is an independent coproduct;
	\item the morphism $\spp$ is unique up to dagger isomorphism;
	\item the morphism $\sem p$ is a dagger monomorphism;
	\item the morphism $\sem s$ is a dagger isomorphism.
\end{itemize}
\begin{proof}
	The second point is a direct consequence of Lemma~\ref{lem:ic-choice}.
	We prove all other three statements by induction on the typing derivations.
	Note that the two first statements only apply to patterns, while the last
	one only applies to unitaries.
	\begin{itemize}
		\item Case $\Phase \theta$. Direct.
		\item Case $s;t$. The induction hypothesis gives that $\sem s$ and
			$\sem t$ are dagger isomorphisms. By definition, we have $\sem{s;t}
			= \sem t \sem s$, and dagger isomorphisms are stable by
			composition.
		\item Case $s \otimes t$. We have by definition $\sem{s \otimes t} =
			\sem s \otimes \sem t$, and $\otimes$ preserves dagger
			isomorphisms.
		\item Case $\id n$. Direct.
		\item Case $\ifthen p s$. We write $u$ for $\sem{\ifthen p s}$.
			Lemma~\ref{lem:univ-monic} ensures that $u$ is dagger monic.
			Additionally, since $(\sem p, \sem p^\perp)$ is an independent
			coproduct, we have unique $u$ and $v$ such that $u \sem p = \sem p
			\sem s$, $u \spp = \spp$, $v \sem p = \sem p \sem s\dg$ and $v \spp
			= \spp$.
			\[
				\begin{tikzcd}[ampersand replacement=\&, row sep=large, column sep=1.5cm]
					\qsem n \& \qsem m \& \bullet \\
					\qsem n \& \qsem m \& \bullet
					\arrow["\sem p", from=1-1, to=1-2]
					\arrow["\sem p"', from=2-1, to=2-2]
					\arrow["\spp"', from=1-3, to=1-2]
					\arrow["\spp", from=2-3, to=2-2]
					\arrow["\sem s"', from=1-1, to=2-1]
					\arrow["\iid", from=1-3, to=2-3]
					\arrow["{u}"{description}, shift right=5pt, dashed, from=1-2, to=2-2]
					\arrow["{v}"{description}, shift right=5pt, dashed, from=2-2, to=1-2]
				\end{tikzcd}
			\]
			We have:
			\[
				\begin{array}{l}
					uv \sem p
					= u \sem p \sem s \dg
					= \sem p \sem s \sem s \dg
					= \sem p \text, \\
					uv \spp
					= u \spp
					= \spp \text.
				\end{array}
			\]
			Since $(\sem p, \spp)$ is jointly epic, we have $uv = \iid$,
			therefore $u$ is a dagger isomorphism.
		\item Case $\ket x$. Direct, by definition.
		\item Case $p \cdot q$. Dagger monomorphisms are preserved by
			composition. We define
			\[
				\sem{p \cdot q}^\perp = \left[\sem p \sem
				q^\perp, \sem p^\perp \right] ~\circ \cong \text.
			\]
			and we wish to show that $(\sem{p \cdot q}, \sem{p \cdot q}^\perp)$
			is an independent coproduct. First, if we have an independent
			cospan $(h_1, h_2)$, then $u$ exists because $(\sem p, \spp)$ is an
			independent coproduct and $v$ exists because $(\sem q, \sem
			q^\perp)$ is an independent coproduct, such that the following
			diagram commutes.
			\[
				\begin{tikzcd}[ampersand replacement=\&, row sep=large, column sep=1.5cm]
					\qsem l \& \qsem n \& \bullet \& \bullet'' \\
					\& \qsem n \& \qsem m \& \bullet' \\
					\& \& \qsem m \&
					\arrow["\sem q", from=1-1, to=1-2]
					\arrow["\sem q^\perp"', from=1-3, to=1-2]
					\arrow["\iid"', from=1-2, to=2-2]
					\arrow["\shortmid"{marking}, tail, from=1-3, to=1-4]
					\arrow["\shortmid"{marking}, tail, from=2-4, to=1-4]
					\arrow["\sem p", from=2-2, to=2-3]
					\arrow["\sem p^\perp"', from=2-4, to=2-3]
					\arrow["{\sem{p \cdot q}^\perp}"{description}, dashed, from=1-4, to=2-3]
					\arrow["{h_1}"{description}, curve={height=30pt}, from=1-1, to=3-3]
					\arrow["{h_2}"{description, pos=0.7}, curve={height=-10pt}, from=1-4, to=3-3]
					\arrow["{u}"{description, pos=0.7}, curve={height=15pt}, dashed, from=1-2, to=3-3]
					\arrow["{v}"{description}, dashed, from=2-3, to=3-3]
				\end{tikzcd}
			\]
			Therefore, there exists a unique mediating morphism.
			Now, if we have $f$ and $g$ such that $f\sem{p \cdot q} = g\sem{p
			\cdot q}$ and $f \sem{p \cdot q}^\perp = g \sem{p \cdot q}^\perp$, then
			\[
				\left\{
					\begin{array}{l}
						f \sem p \sem q = g \sem p \sem q \\
						f \left[\sem p \sem q^\perp, \sem p^\perp \right]
						=
						g \left[\sem p \sem q^\perp, \sem p^\perp \right] \text,
					\end{array}
				\right.
			\]
			thus
			\[
				\left\{
					\begin{array}{l}
						f \sem p \sem q = g \sem p \sem q \\
						f \sem p \sem q^\perp = g \sem p \sem q^\perp \\
						f \sem q^\perp = g \sem q^\perp
					\end{array}
				\right.
			\]
			The two first lines imply that $f \sem p = g \sem p$ because $(\sem
			q, \sem q^\perp)$ is jointly epic, which implies with the last line
			that $f = g$ because $(\sem p, \spp)$ is jointly epic. Thus $(\sem{p
			\cdot q}, \sem{p \cdot q}^\perp)$ is jointly epic. We conclude that
			$(\sem{p \cdot q}, \sem{p \cdot q}^\perp)$ is an independent
			coproduct.
		\item Case $p \otimes q$ is very similar to the previous one, with
			\[
				\sem{p \otimes q}^\perp = ~\cong\circ \left[\sem p^\perp
				\otimes \sem q, \left[ \sem p \otimes \sem q^\perp, \sem
				p^\perp \otimes \sem q^\perp \right] \right] \circ\cong
			\]
	\end{itemize}
\end{proof}

\subsection{Semantics of ``if let''}

We prove Proposition~\ref{prop:equations} in separate lemmas.

\begin{lemma}
	We have that
	\(
		\sem{\ifthen p {\ifthen q s}} = \sem{\ifthen {p \cdot q} s}.
	\)
\end{lemma}
\begin{proof}
	The morphism $\sem{\ifthen {p \cdot q} s}$ is, by definition, the unique
	$u$ such that
	\[
		\left\{
			\begin{array}{l}
				u\sem{p \cdot q} = \sem p \sem q \sem s \\
				u \sem{p \cdot q}^\perp = \sem{p \cdot q}^\perp
			\end{array}
		\right.
	\]
	which means that it is the unique $u$ such that
	\[
		\left\{
			\begin{array}{l}
				u \sem p \sem q = \sem p \sem q \sem s \\
				u \sem p \sem q^\perp = \sem p \sem q^\perp \\
				u \sem p^\perp = \sem p^\perp
			\end{array}
		\right.
	\]
	Moreover, $\sem{\ifthen q s}$ is, by definition, the unique $v$ such that
	\[
		\left\{
			\begin{array}{l}
				v\sem q = \sem q \sem s \\
				v \sem q^\perp = \sem q^\perp
			\end{array}
		\right.
	\]
	and $\sem{\ifthen p {\ifthen q s}}$ is, by definition, the unique $w$ such
	that
	\[
		\left\{
			\begin{array}{l}
				w\sem p = \sem p v \\
				w \sem p^\perp = \sem p^\perp
			\end{array}
		\right.
	\]
	If we precompose the first line of the last observation by $\sem q$ (resp.
	$\sem q^\perp$), we then obtain the following:
	\[
		\left\{
			\begin{array}{l}
				w \sem p \sem q = \sem p \sem q \sem s \\
				w \sem p \sem q^\perp = \sem p \sem q^\perp \\
				w \sem p^\perp = \sem p^\perp
			\end{array}
		\right.
	\]
	We conclude that $u = w$ by uniqueness.
\end{proof}

\begin{lemma}
	We have that
	\(
		\sem{\ifthen p {(s;t)}} = \sem{\ifthen p s ; \ifthen p t}.
	\)
\end{lemma}
\begin{proof}
	By uniqueness.
	\[
		\begin{tikzcd}[column sep=2cm]
			\ & \ & \ \\
			\ & \ & \ \\
			\ & \ & \
			% \arrow["0"', bend left=30, from=1-1, to=1-3]
			\arrow["\sem p"', from=1-1, to=1-2]
			\arrow["\sem p"', from=2-1, to=2-2]
			\arrow["\sem p"', from=3-1, to=3-2]
			\arrow["\sem p^\perp", from=1-3, to=1-2]
			\arrow["\sem p^\perp", from=2-3, to=2-2]
			\arrow["\sem p^\perp", from=3-3, to=3-2]
			\arrow["\sem s"', from=1-1, to=2-1]
			\arrow["\sem t"', from=2-1, to=3-1]
			\arrow["\iid", from=1-3, to=2-3]
			\arrow["\iid", from=2-3, to=3-3]
			\arrow[dashed, from=1-2, to=2-2]
			\arrow[dashed, from=2-2, to=3-2]
			\arrow[dashed, curve={height=20pt}, from=1-2, to=3-2]
		\end{tikzcd}
	\]
\end{proof}

\begin{lemma}
	We have that
	\(
		\sem{\ifthen t s} = \sem t \sem s \sem t\dg.
	\)
\end{lemma}
\begin{proof}
	The morphism $\sem{\ifthen t s}$ is, by definition, the unique $u$ such that
	\[
		\left\{
			\begin{array}{l}
				u \sem t = \sem t \sem s \\
				u 0 = 0
			\end{array}
		\right.
	\]
	Since the second equation always holds, the morphism $u$ is the unique one
	such that $u \sem t = \sem t \sem s$, which also holds for $\sem t \sem s
	\sem t\dg$, thus we conclude that $u = \sem t \sem s \sem t\dg$.
\end{proof}

\subsection{Proof of soundness}\label{app:soundness}

Since $q$ designates the subspace on which the program is applied, we should
have that if $\evalu {q,l,r} s = c$, then $\sem c$ is the mediating morphism $u$
in:
\begin{equation}\label{eq:eval-sem}
	\begin{tikzcd}[ampersand replacement=\&, column sep=1.5cm]
		\qsem m \& \qsem n \& \ \bullet\  \\
		\qsem l \otimes \qsem k \otimes \qsem r  \&\& \\
		\qsem l \otimes \qsem k \otimes \qsem r  \&\& \\
		\qsem m \& \qsem n \& \ \bullet\
		\arrow["\sem q"', "\shortmid"{marking}, tail, from=1-1, to=1-2]
		% \arrow["0"', bend left=30, from=1-1, to=1-3]
		\arrow["u"{description}, dashed, from=1-2, to=4-2]
		\arrow["\ {\sem q^\perp}", "\shortmid"{marking}, tail, from=1-3, to=1-2]
		\arrow["\sem q"', from=4-1, to=4-2]
		\arrow["{\sem q^\perp}", from=4-3, to=4-2]
		\arrow["\xi\dg"', from=1-1, to=2-1]
		\arrow["\iid \otimes \sem s \otimes \iid"', from=2-1, to=3-1]
		\arrow["\xi"', from=3-1, to=4-1]
		\arrow["\iid", from=1-3, to=4-3]
	\end{tikzcd}
\end{equation}
where $\xi \colon \qsem l \otimes \qsem k \otimes \qsem r \to \qsem m$ is the
coherence isomorphism between the two objects, since we know that $l + k + r =
m$.

We have the following:
\begin{itemize}
	\item If $\evalu {q,l,r} s = c$, then $\sem c$ is the mediating morphism
		$u$ of the diagram above;
	\item if $\evalp {q,l,r} p = (c,q')$, then
		\(
			\sem c \sem{q'} = \sem q \xi (\iid \otimes \sem p \otimes \iid) {\xi'}\dg \text.
		\)
		Note that it implies that $\sem{c \cdot q'} = \sem{q \cdot (\id l \otimes p \otimes \id r)}$,
		and therefore
		\begin{align*}
			\sem c \sem{q'}^\perp
			&= [0, \sem c \sem{q'}^\perp]
			= \sem{c \cdot q'}^\perp \\
			&= \sem{q \cdot (\id l \otimes p \otimes \id r)}^\perp \\
			&= \left[ \sem q ^\perp , \sem q \xi (\iid \otimes \spp \otimes
			\iid) \right] \circ {\cong}\text.
		\end{align*}
\end{itemize}

\begin{proof}
	\begin{mathpar}
		\inferrule{ }{\evalu {q,l,r} {\Phase{\theta}} = [\ifthen q {\Phase{\theta} \otimes \id{l+r}}]}\and
		\inferrule{\evalu {q,l,r} s = c \and \evalu {q,l,r} t = c'}{\evalu {q,l,r} {s ; t} = c \doubleplus c'}\and
		\inferrule{\vdash s : \unitary{k_1} \and \vdash t : \unitary{k_2} \and \evalu {q,l,r+k_2} s = c \and \evalu {q,l+k_1,r} t = c'}{\evalu {q,l,r} {s \otimes t} = c \doubleplus c'}\and
		\inferrule{ }{\evalu {q,l,r} {\id{k}} = []}\and
		\inferrule{\evalp {q,l,r} p = (c, q') \and \evalu {q', l, r} s = c'}{\evalu {q,l,r} {\ifthen p s} = c^\dagger \doubleplus c' \doubleplus c}\and
		\inferrule{x \in \{0, 1, +, -\}}{\evalp {q,l,r} {\ket{x}} = ([], q[\ket{x}/l])}\and
		\inferrule{\evalu {q, l, r} s = c} {\evalp {q,l, r} s = (c, q)}\and
		\inferrule{\evalp {q,l,r} {p_1} = (c, q') \and \evalp {q', l, r} {p_2} = (c', q'')}{\evalp {q,l,r} {p_1 \cdot p_2} = (c' \doubleplus c, q'')}\and
		\inferrule{\vdash p_1 : \pattern {j_1} {k_1} \and \vdash p_2 : \pattern {j_2} {k_2} \and \evalp {q,l,r + k_2} {p_1} = (c, q') \and \evalp {q', l + j_1, r} {p_2} = (c', q'')}{\evalp {q, l, r} {p_1 \otimes p_2} = (c' \doubleplus c, q'')}
	\end{mathpar}
	\begin{itemize}
		\item Case $\evalu {q,l,r} {\Phase{\theta}}$. This one is true by definition.
		\item Case $\evalu {q,l,r} {s ; t}$. Direct by composing the two diagrams
			for the definition of the semantics of the evaluation (see
			(\ref{eq:eval-sem})), because $\xi$ is a dagger isomorphism.
		\item Case $\evalu {q, l,r} {s \otimes t}$.
			Without loss of generality, we will work here with
			\[
				\xi \colon \qsem l \otimes \qsem{k_1} \otimes \qsem {k_2} \otimes \qsem r \to \qsem m
			\]
			and thus,
			\begin{align*}
				&\! \xi (\iid \otimes \iid \otimes \sem t \otimes \iid) \xi\dg
				\xi (\iid \otimes \sem s \otimes \iid \otimes \iid) \xi\dg
				\\
				&= \xi (\iid \otimes \iid \otimes \sem t \otimes \iid)
				(\iid \otimes \sem s \otimes \iid \otimes \iid) \xi\dg
				\\
				&= \xi (\iid \otimes \sem s \otimes \sem t \otimes \iid) \xi\dg
			\end{align*}
			which is enough to conclude.
		\item Case $\evalu {q,l,r} {\id{k}}$ is trivial.
		\item Case $\evalu {q,l,r} {\ifthen p s}$. Let $u_1$ denote the
			semantics of $c$, and $u_2$ the semantics of $c'$ and $u$ the
			semantics of $\ifthen p s$. We wish to prove that $u_1 u_2 u_1\dg$
			is the mediating morphism in the following diagram:
			
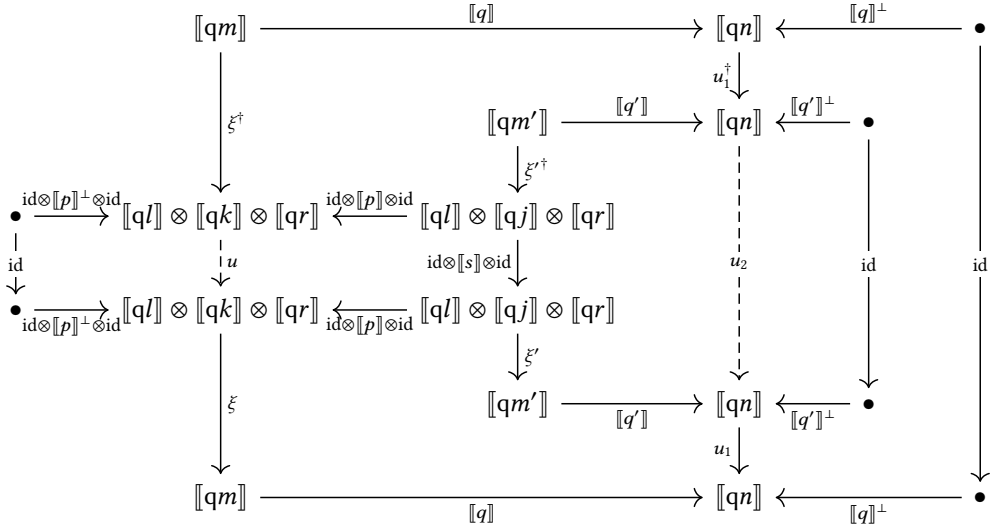
\[
				\begin{tikzcd}[column sep=1cm]
					& {\qsem m} && {\qsem n} && \bullet \\
					&& {\qsem{m'}} & {\qsem n} & \bullet \\
					\bullet & {\qsem l \otimes \qsem k \otimes \qsem r} &
					{\qsem l \otimes \qsem j \otimes \qsem r} \\
					\bullet & {\qsem l \otimes \qsem k \otimes \qsem r} &
					{\qsem l \otimes \qsem j \otimes \qsem r} \\
					&& {\qsem{m'}} & {\qsem n} & \bullet \\
					& {\qsem m} && {\qsem n} && \bullet
					\arrow["{\sem q}", from=1-2, to=1-4]
					\arrow["\xi\dg", from=1-2, to=3-2]
					\arrow["{u_1\dg}"', from=1-4, to=2-4]
					\arrow["{\sem q ^\perp}"', from=1-6, to=1-4]
					\arrow["\iid"{description}, from=1-6, to=6-6]
					\arrow["{\sem{q'}}", from=2-3, to=2-4]
					\arrow["{\xi'}\dg", from=2-3, to=3-3]
					\arrow["{u_2}"{description}, dashed, from=2-4, to=5-4]
					\arrow["{\sem{q'}^\perp}"', from=2-5, to=2-4]
					\arrow["\iid"{description}, from=2-5, to=5-5]
					\arrow["{\iid \otimes \sem{p}^\perp \otimes \iid}", from=3-1, to=3-2]
					\arrow["\iid"{description}, from=3-1, to=4-1]
					\arrow["u", dashed, from=3-2, to=4-2]
					\arrow["{\iid \otimes \sem p \otimes \iid}"', from=3-3, to=3-2]
					\arrow["{\iid \otimes \sem s \otimes \iid}"', from=3-3, to=4-3]
					\arrow["{\iid \otimes \sem{p}^\perp \otimes \iid}"', from=4-1, to=4-2]
					\arrow["\xi", from=4-2, to=6-2]
					\arrow["{\iid \otimes \sem p \otimes \iid}", from=4-3, to=4-2]
					\arrow["\xi'", from=4-3, to=5-3]
					\arrow["{\sem{q'}}"', from=5-3, to=5-4]
					\arrow["{u_1}"', from=5-4, to=6-4]
					\arrow["{\sem{q'}^\perp}", from=5-5, to=5-4]
					\arrow["{\sem{q}}"', from=6-2, to=6-4]
					\arrow["{\sem{q}^\perp}", from=6-6, to=6-4]
				\end{tikzcd}
			\]
			First, we know by induction hypothesis that
			\[
				u_1 \sem{q'}^\perp = \left[ \sem q ^\perp ,
				\sem q \xi (\iid \otimes \spp \otimes \iid) \right] \circ {\cong}
			\]
			which means that
			\[
				\left\{
				\begin{array}{l}
					\sem{q'}^\perp {\cong} i_1 = u_1\dg \sem q ^\perp \\
					\sem{q'}^\perp {\cong} i_2 = u_1\dg \sem q \xi (\iid \otimes \spp \otimes \iid)
				\end{array}
				\right.
			\]
			and we also know that:
			\[
				u_1 \sem{q'} = \sem q \xi (\iid \otimes \sem p \otimes \iid) {\xi'}\dg
			\]
			and that $u_2$ is the mediating morphism for the right internal square.
			We need to prove that $u_1 u_2 u_1\dg$ is indeed the mediating
			morphism of the diagram, meaning that:
			\[
				\left\{
				\begin{array}{l}
					u_1 u_2 u_1\dg \sem q = \sem q \xi u \xi \dg \\
					u_1 u_2 u_1\dg \sem q ^\perp = \sem q ^\perp
				\end{array}
				\right.
			\]
			The first equation is equivalent to:
			\[
				u_2 u_1\dg \sem q \xi = u_1\dg \sem q \xi u
			\]
			and since $(\iid \otimes \sem p \otimes \iid, \iid \otimes \spp
			\otimes \iid)$ is jointly epic, it suffices to prove that
			\[
				\left\{
				\begin{array}{l}
					u_2 u_1\dg \sem q \xi (\iid \otimes \sem p \otimes \iid) =
					u_1\dg \sem q \xi u (\iid \otimes \sem p \otimes \iid) \\
					u_2 u_1\dg \sem q \xi (\iid \otimes \spp \otimes \iid) =
					u_1\dg \sem q \xi u (\iid \otimes \spp \otimes \iid)
				\end{array}
				\right.
			\]
			Now,
			\begin{align*}
				u_2 u_1\dg \sem q \xi (\iid \otimes \sem p \otimes \iid)
				&= u_2 \sem{q'} \xi' \\
				&= \sem{q'} {\xi'} (\iid \otimes \sem s \otimes \iid) \\
				&= u_1\dg \sem q \xi (\iid \otimes \sem p \otimes \iid) (\iid \otimes \sem s \otimes \iid) \\
				&= u_1\dg \sem q \xi u (\iid \otimes \sem p \otimes \iid) \\[1.5ex]
				u_2 u_1\dg \sem q \xi (\iid \otimes \spp \otimes \iid)
				&= u_2 \sem{q'}^\perp {\cong} i_2 \\
				&= \sem{q'}^\perp {\cong} i_2 \\
				&= u_1\dg \sem q \xi (\iid \otimes \spp \otimes \iid)
			\end{align*}
			and finally,
			\begin{align*}
				u_2 u_1\dg \sem q ^\perp
				&= u_2 \sem{q'}^\perp {\cong} i_1 \\
				&= \sem{q'}^\perp {\cong} i_1 \\
				&= u_1\dg \sem q ^\perp
			\end{align*}
			which concludes.
		\item Case $\evalp {q,l,r} {\ket{x}}$. Direct.
		\item Case $\evalp {q,l,r} s$. Let $u$ be the mediating morphism in the
			induction hypothesis. We then have
			\begin{align*}
				u \sem{q}
				&= \sem q \xi (\iid \otimes \sem s \otimes \iid) \xi\dg
			\end{align*}
			which concludes.
		\item Case $\evalp {q,f} {p_1 \cdot p_2}$. By induction hypothesis, we
			know that:
			\[
				u_1 \sem{q'} = \sem q \xi (\iid \otimes \sem{p_1} \otimes \iid) {\xi'}\dg
			\]
			and
			\[
				u_2 \sem{q''} = \sem{q'} \xi' (\iid \otimes \sem{p_2} \otimes \iid) {\xi''}\dg
			\]
			Thus,
			\begin{align*}
				u_1 u_2 \sem{q''}
				&= u_1 \sem{q'} \xi' (\iid \otimes \sem p_2 \otimes \iid) {\xi''}\dg \\
				&= \sem q \xi (\iid \otimes \sem{p_1} \otimes \iid) {x'}\dg
				\xi' (\iid \otimes \sem{p_2} \otimes \iid) {\xi''}\dg \\
				&= \sem q \xi (\iid \otimes \sem{p_1} \sem{p_2} \otimes \iid)
				{\xi''}\dg \\
				&= \sem q \xi (\iid \otimes \sem{p_1 \cdot p_2} \otimes \iid)
				{\xi''}\dg
			\end{align*}
			which concludes.
		\item Case $\evalp {q,l,r} {p_1 \otimes p_2}$. Without loss of
			generality, we will work with
			\[
				\begin{array}{l}
					\xi_1 \colon \qsem l \otimes \qsem{k_1} \otimes \qsem {k_2}
					\otimes \qsem r \to \qsem m \text, \\
					\xi_2 \colon \qsem l \otimes \qsem{j_1} \otimes \qsem {k_2}
					\otimes \qsem r \to \qsem m'  \text, \\
					\xi_3 \colon \qsem l \otimes \qsem{j_1} \otimes \qsem {j_2}
					\otimes \qsem r \to \qsem m  \text.
				\end{array}
			\]
			By induction hypothesis, we know that:
			\[
				u_1 \sem{q'} = \sem q \xi_1 (\iid \otimes \sem{p_1} \otimes
				\iid \otimes \iid) {\xi_2}\dg
			\]
			and
			\[
				u_2 \sem{q''} = \sem{q'} \xi_2 (\iid \otimes \iid \otimes
				\sem{p_2} \otimes \iid) {\xi_3}\dg
			\]
			Thus,
			\begin{align*}
				u_1 u_2 \sem{q''}
				&= u_1 \sem{q'} \xi_2 (\iid \otimes \iid \otimes \sem{p_2}
				\otimes \iid) {\xi_3}\dg \\
				&= \sem q \xi_1 (\iid \otimes \sem{p_1} \otimes \iid \otimes
				\iid) {\xi_2}\dg \xi_2 (\iid \otimes \iid \otimes \sem{p_2}
				\otimes \iid) {\xi_3}\dg \\
				&= \sem q \xi_1 (\iid \otimes \sem{p_1} \otimes \iid \otimes
				\iid) (\iid \otimes \iid \otimes \sem{p_2} \otimes \iid)
				{\xi_3}\dg \\
				&= \sem q \xi_1 (\iid \otimes \sem{p_1} \otimes \sem{p_2}
				\otimes \iid) {\xi_3}\dg  \\
				&= \sem q \xi_1 (\iid \otimes \sem{p_1 \otimes p_2}
				\otimes \iid) {\xi_3}\dg
			\end{align*}
			which concludes.
	\end{itemize}
\end{proof}

Soundness is then simply a corollary.

\end{document}